\documentclass[11pt]{article}
\usepackage[letterpaper,hmargin=1in,vmargin=1.25in]{geometry} 
\usepackage{microtype}
\usepackage{graphicx}                   
\usepackage{amsmath}                    
\usepackage{amssymb}                    
\usepackage{amsfonts}                   
\usepackage{mathrsfs}                   
\usepackage{url}                        
\usepackage{color}                      
\usepackage{cite}                       
\usepackage{hyperref}                   
\usepackage{enumitem}                   
\hypersetup{colorlinks=true,linkcolor=[rgb]{0.75,0,0},citecolor=[rgb]{0,0,0.75}}
\usepackage{amsthm}                     
\usepackage{thmtools}                    
\usepackage{thm-restate}                

\declaretheorem[name=Theorem,numberwithin=section]{theorem}
\declaretheorem[name=Lemma,sibling=theorem]{lemma}
\declaretheorem[name=Corollary,sibling=lemma]{corollary}

\newcommand{\floor}[1]{\left\lfloor #1\right\rfloor}

\newcommand{\ang}[1]{\langle #1\rangle}

\newcommand{\inv}[1]{\frac{1}{#1}}
\renewcommand{\ang}[1]{\langle #1\rangle}

\newcommand{\RE}{\mathbb{R}}            
\newcommand{\eps}{\varepsilon}          
\newcommand{\RR}{\mathscr{R}}
\renewcommand{\AA}{\mathscr{A}}
\newcommand{\CC}{\mathscr{C}}
\newcommand{\EE}{\mathscr{E}}
\newcommand{\WW}{\mathscr{W}}

\newcommand{\bd}{\partial}
\newcommand{\etal}{\textit{et al.}}

\DeclareMathOperator{\vol}{vol}
\DeclareMathOperator{\area}{area}
\DeclareMathOperator{\conv}{conv}
\DeclareMathOperator{\radius}{radius}
\DeclareMathOperator{\width}{width}
\DeclareMathOperator{\ray}{ray}

\DeclareMathOperator{\base}{base}

\DeclareMathOperator{\polylog}{polylog}

\begin{document}
\title{Optimal Bound on the Combinatorial Complexity of Approximating Polytopes}

\author{%
    Rahul Arya\\
		Department of Electrical Engineering and Computer Science \\
		University of California, Berkeley, California \\
		rahularya@berkeley.edu \\
		\and
	Sunil Arya\thanks{Research supported by the Research Grants Council of Hong Kong, China under project number 16214518. The work of David Mount was supported by NSF grant CCF--1618866. The work of Guilherme da Fonseca has been supported by the French ANR PRC grant ADDS (ANR-19-CE48-0005).}\\
		Department of Computer Science and Engineering \\
		The Hong Kong University of Science and Technology, Hong Kong\\
		arya@cse.ust.hk \\
		\and
	Guilherme D. da Fonseca\footnotemark[1]\\
	    Aix-Marseille Universit\'{e} and LIS, France\\
		guilherme.fonseca@lis-lab.fr
		\and
	David M. Mount\footnotemark[1]\\
		Department of Computer Science and 
		Institute for Advanced Computer Studies \\
		University of Maryland, College Park, Maryland \\
		mount@umd.edu \\
}
\date{}

\maketitle

\begin{abstract}
This paper considers the question of how to succinctly approximate a multidimensional convex body by a polytope. Given a convex body $K$ of unit diameter in Euclidean $d$-dimensional space (where $d$ is a constant) and an error parameter $\varepsilon > 0$, the objective is to determine a convex polytope of low combinatorial complexity whose Hausdorff distance from $K$ is at most $\varepsilon$. By \emph{combinatorial complexity} we mean the total number of faces of all dimensions. Classical constructions by Dudley and Bronshteyn/Ivanov show that $O(1/\varepsilon^{(d-1)/2})$ facets or vertices are possible, respectively, but neither achieves both bounds simultaneously. In this paper, we show that it is possible to construct a polytope with $O(1/\varepsilon^{(d-1)/2})$ combinatorial complexity, which is optimal in the worst case.

Our result is based on a new relationship between $\varepsilon$-width caps of a convex body and its polar body. Using this relationship, we are able to obtain a volume-sensitive bound on the number of approximating caps that are ``essentially different.'' We achieve our main result by combining this with a variant of the witness-collector method and a novel variable-thickness layered construction of the economical cap covering.
\end{abstract}

\section{Introduction} \label{s:intro}

Convex objects are of central importance in numerous areas of geometric computation. Efficiently approximating a multi-dimensional convex body by a convex polytope is a natural and fundamental problem. Given a closed, convex set $K$ of unit diameter in Euclidean $d$-dimensional space and an error parameter $\eps > 0$, the objective is to produce a convex polytope of low combinatorial complexity whose Hausdorff distance%
\footnote{The \emph{Hausdorff distance} between any two sets is the maximum Euclidean distance between any point in one set and its closest point in the other set. While there are other metrics for polytope similarity (see, e.g.~\cite{Bro08}), Hausdorff is the measure most often used in computational geometry. Approximations sensitive to the diameter and the directional width can be obtained by applying an affine transformation to $K$.}
from $K$ is at most $\eps$. The \emph{combinatorial complexity} of a polytope is the total number of faces of all dimensions. Throughout, we assume that the dimension $d$ is a constant.

The bounds given in the literature for convex approximation are of two common types~\cite{Bro08}. In both cases, the bounds hold for all $\eps \leq \eps_0$, for some $\eps_0 > 0$. In \emph{nonuniform bounds}, the value of $\eps_0$ depends on $K$. Examples include bounds by Gruber~\cite{Gru93}, B{\" o}r{\" o}czky \cite{Bor00a, Bor00b}, B\'{a}r\'{a}ny~\cite{Bar89}, Clarkson~\cite{Cla06}, and others \cite{ANS16, Sch87, Tot48}. Of particular relevance here is the work of B{\" o}r{\" o}czky \cite{Bor00b}, whose results imply that any unit-diameter convex body with sufficiently smooth boundary can be $\eps$-approximated with total complexity $O(1/\eps^{(d-1)/2})$. But this holds for $\eps \leq \eps_0$, where $\eps_0$ depends on the curvature properties of the body. In contrast, our interest is in \emph{uniform bounds}, where the value of $\eps_0$ does not depend on $K$ (subject to the assumption of unit diameter). Examples of such bounds include the results of Dudley~\cite{Dud74}, Bronshteyn and Ivanov~\cite{BrI76}, and our own prior work~\cite{AFM12b,AFM17a}. Such bounds hold without any assumptions on $K$.

Dudley showed that, for $\eps \leq 1$, any convex body $K$ of unit diameter can be $\eps$-approximated by a convex polytope $P$ with $O(1/\eps^{(d-1)/2})$ facets~\cite{Dud74}. This bound is known to be tight in the worst case and is achieved when $K$ is a Euclidean ball~\cite{Bro08}. Alternatively, Bronshteyn and Ivanov showed the same bound holds for the number of vertices, which is also the best possible~\cite{BrI76}. Similar bounds are widely used in algorithms based on $\eps$-kernels to approximate the diameter, width, minimum enclosing cylinder, and bichromatic closest pair, among others (see~\cite{AHV05,AFM17c,Cha18}).
Unfortunately, no construction is known that matches both bounds simultaneously. This issue has been noted by Clarkson~\cite{Cla06}, where he cites communications with Jeff Erickson showing that both bounds can be attained but at the cost of sacrificing convexity. 

McMullen's \emph{Upper-bound Theorem}~\cite{MCM70} implies that a polytope with $n$ facets (resp., vertices) has $O(n^{\floor{d/2}})$ vertices (resp., facets), and this bound is attained by cyclic polytopes. Applying this to Dudley's or Bronshteyn and Ivanov's constructions yields a very weak upper bound of roughly $O\big( 1/\eps^{(d^2-d)/4} \big)$ on the combinatorial complexity of $\eps$-approximating polytopes. (Alternative constructions are known that yield a complexity of roughly $O(1/\eps^{d - 2})$~\cite{And63,Bar08}, but this is nearly quadratic in the lower bound.) 

Because it is often useful to convert between vertex-based and facet-based representations of convex polytopes (as the convex hull of points and the intersection of halfspaces, respectively), this blowup has been a major impediment to the application of fundamental polytopal structures such as convex hulls, Delaunay triangulations, and Voronoi diagrams in dimensions $d > 3$. Efficiently representing the combinatorial structure of a polytope's faces is of great practical importance to several algorithms~\cite{BKT17}. However, the high combinatorial complexity of existing polytope approximations severely limits the efficiency of such data structures.

In this paper we resolve this decades-old problem. We present a construction for approximating a convex body that not only simultaneously achieves the bound of $O(1/\eps^{(d-1)/2})$ on the number of vertices and facets, but in fact establishes this bound on the total combinatorial complexity (sum of faces of all dimensions).

\begin{theorem} \label{thm:main}
Let $K \subset \RE^d$ be a convex body of unit diameter, where $d$ is a fixed constant. For all sufficiently small positive $\eps$ (independent of $K$) there exists an $\eps$-approximating convex polytope $P$ to $K$ of combinatorial complexity $O(1/\eps^{(d-1)/2})$.
\end{theorem}

We considered this problem earlier~\cite{AFM17a}, obtaining a result that was suboptimal. That paper introduced two useful techniques: a width-based variant of B{\'a}r{\'a}ny's~\cite{Bar00} economical cap cover based on Macbeath regions and a multi-layered%
\footnote{Note that the term ``layer'' also appears in \cite{DGGT16} with respect to the witness-collector method, but our usage is entirely different.}
approach to the witness-collector method~\cite{DGGT16}. That result fell short, however, due in part to a weak understanding of the distribution of Macbeath region volumes in the cap cover. In this paper, we introduce an important new result, namely a ``volume-sensitive'' bound on the number of Macbeath regions (in Theorem~\ref{thm:vol-sensitive-bound}).

This new bound comes about by establishing a correspondence between Macbeath regions in the original convex body and its polar body, and then demonstrating a reciprocal relationship in the volumes of corresponding regions. A classical result in the theory of convex bodies states that the volume of a convex body and its polar dual have a reciprocal relationship. The dimensionless product of these two quantities is called the \emph{Mahler volume}~\cite{Kuperberg}. Our correspondence shows that this property holds not just globally, but locally as well. This enables this reciprocal relationship to be applied in the context of approximation, and this admits a more sophisticated application of the witness-collector method. We believe that this primal-polar approach is an important new technique, which will be useful in other optimization problems involving convex approximation.

\subsection*{Overview of Methods} 

Before delving into technical matters, let us survey the broader context behind our work and give a high-level view of our approach. Convex approximation by polytopes is in essence a covering problem. Given a convex body $K$ of unit diameter, an \emph{$\eps$-width cap} is the intersection of $K$ with a halfspace that cuts off a slice of width $\eps$ from $K$. Clearly, any collection of $\eps$-width caps that covers all of $K$'s boundary yields an $\eps$-approximation of $K$ having (at most) as many facets. Clarkson~\cite{Cla06} observed that, as $\eps$ tends to zero, computing such a cover involves sampling the boundary of $K$ according to a metric that is sensitive to $K$'s shape, with proportionately more samples in areas of higher curvature. Intuitively, such a metric should capture the notion of the ``local feature size'' at any point of $K$.

In recent works, we have demonstrated that shape-sensitive sampling can be achieved through the use of Macbeath regions. Given a point $x \in K$, the \emph{Macbeath region} $M(x)$ is a maximal centrally symmetric convex shape centered at $x$ and contained within $K$. (Formal definitions and properties are provided in Section~\ref{s:prelim-macbeath}.) Macbeath regions enjoy many useful properties. They can be computed efficiently, they have nice packing and covering properties, and up to constant scaling factors, $M(x)$ approximates the minimum-volume cap centered at $x$ as well as the unit balls centered at $x$ in both the Hilbert and Blaschke geometries induced by $K$~\cite{Bar00,VeW16,Tho17}. 

Macbeath regions were first introduced to computational geometry as a tool to prove lower bounds for range searching~\cite{BCP93,AMX12}. Later on, they were used to prove existential results in convex approximation~\cite{AFM17a,AFM17b,DGJ19,MuR14}. More recently, the explicit computation of such regions has been used to obtain the fastest algorithms known for several approximation problems such as $\eps$-kernel, diameter, and width~\cite{AFM17c,AFM18}.

In spite of their obvious relevance to convex approximation, there is still much that is not known about Macbeath regions. In the context of convex approximation, the number of disjoint (shrunken) Macbeath regions that can be placed at distance $\eps$ of $K$'s boundary is closely related to the complexity of approximating $K$. In earlier work~\cite{AFM17a}, we showed that for any convex body of unit diameter, such a set has size $O(1/\eps^{(d-1)/2})$. This bound is attained for a ball, which has $\Theta(1/\eps^{(d-1)/2})$ disjoint such regions, all of volume $\Theta(\eps^{(d+1)/2})$ (see Figure~\ref{f:small-macbeath}(a)).

The volume distribution of such a set is a question of key importance. To see why this is nontrivial, consider the unit $d$-dimensional hypercube (see Figure~\ref{f:small-macbeath}(b)). Parallel to each facet is an $\eps$-width cap (and Macbeath region) of very large volume $\Theta(\eps)$. Since the volume of the portion of $K$ lying within distance $\eps$ of the boundary is also $\Theta(\eps)$, a packing argument implies that there cannot be more than a constant number of disjoint Macbeath regions associated with such large volume caps. On the other hand, $\eps$-width caps (and Macbeath regions) that are orthogonal to the main diagonals, that is, close to the vertices of the hypercube, have very small volume of $\Theta(\eps^{d})$. A packing argument provides no useful bound on their number. Nonetheless, there cannot be many of them. To see why, observe that a small volume $\eps$-width cap can only exist where $K$'s boundary has high curvature, and the total curvature of any convex body is bounded. In this paper, we formalize and generalize this intuition. In particular, we show (in Theorem~\ref{thm:vol-sensitive-bound}) that the number of disjoint $\eps$-width Macbeath regions of volume $v$ is $O(\min(\eps/v, v/\eps^d))$. Note that in both of the above extremes, this yields a tight bound of $O(1)$ on the number of Macbeath regions, whereas the maximum number of Macbeath regions is attained by the intermediate volume of $\Theta(\eps^{(d+1)/2})$.

\begin{figure*}[btp]
  \centerline{\includegraphics[scale=.65]{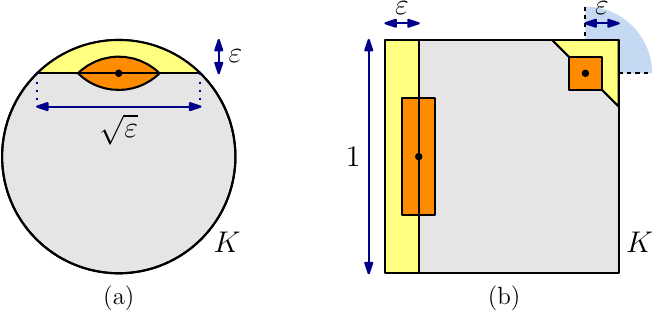}}
  \caption{\label{f:small-macbeath} Macbeath regions of different volumes in a ball and hypercube.}
\end{figure*}

To prove this volume-sensitive bound we explore the nature of the Macbeath regions and caps in the polar body $K^*$ of $K$. We show that for any cap $C$ of width $\eps$ in the primal body $K$, there is a corresponding cap $C'$ of width $\Theta(\eps)$ in the polar body $K^*$ such that the bases of these two caps are rough polars of each other. By the Mahler volume, it follows that the product of the volumes of these corresponding caps is roughly the same for every $\eps$-width cap of $K$. The volume-sensitive bound then follows by applying a packing argument in either $K$ or $K^*$, depending on which cap has the larger volume. 

In order to apply this new volume-sensitive bound to convex approximation, we recall the witness-collector approach to bound the combinatorial complexity of the convex hull of a set of points~\cite{DGGT16,AFM17a}. Let $S \subset \RE^d$ be a set of points. We define a set $\WW$ of regions called \emph{witnesses} and a corresponding set $\CC$ of regions called \emph{collectors}, which satisfy the following properties:
\begin{enumerate}
\item[(1)] Each witness of $\WW$ contains a point of $S$ in its interior.

\item[(2)] Any halfspace $H$ either contains a witness $W \in \WW$ or $H \cap S$ is contained in a collector $C \in \CC$.

\item[(3)] Each collector $C \in \CC$ contains a constant number of points of $S$.
\end{enumerate}
Devillers {\etal}~\cite{DGGT16} showed that given a set of witnesses $\WW$ and collectors $\CC$ satisfying the above properties, the combinatorial complexity of the polytope $K$ defined as the convex hull of $S$ is $O(|\CC|)$. Hence, to prove Theorem~\ref{thm:main}, it suffices to provide a set $S$ whose convex hull $\eps$-approximates $K$ and a corresponding set of witnesses and collectors of cardinality $O(1/\eps^{(d-1)/2})$.

The key tool for our construction is the aforementioned collection of disjoint Macbeath regions lying within distance $\eps$ of $K$'s boundary. This set of Macbeath regions has the desired size, and the corresponding caps satisfy properties (1) and (2) above, where the Macbeath regions represent the witnesses and the caps represent the collectors. Since Macbeath regions approximate caps, if we define a set $S$ by picking one arbitrary point inside each Macbeath region, then we are guaranteed to obtain an $\eps$-approximation of $K$.

However, the construction may fail to satisfy property (3). While a cap may only intersect a constant number of Macbeath regions of larger or similar volume, any given cap may intersect a large number of Macbeath regions of smaller volume. We dealt with this in~\cite{AFM17a} by arranging the Macbeath regions into $O(\log\inv\eps)$ \emph{layers} of thickness $\eps$, moving low-volume Macbeath regions into the innermost layers. In this way, we prevented the caps from intersecting Macbeath regions of smaller volume, and assuring property (3). However, the approximation error grew from $\eps$ to $O(\eps \log\inv\eps)$ because of the increased total thickness of the layers. An $\eps$-approximation was achieved through a compensatory scaling of $\eps$, causing the complexity to grow to $O((\log\inv\eps/\eps)^{(d-1)/2})$.

In this paper we show how to exploit our volume-sensitive bound to obtain an $\eps$-approximation of optimal combinatorial complexity. As in~\cite{AFM17a}, we place Macbeath regions in different layers of thickness according to their volumes, the outermost layers corresponding to Macbeath regions of larger volume and the innermost layers corresponding to smaller volumes. A key idea is to abandon the use of layers of uniform thickness and instead use layers of varying thicknesses. The middle layer (numbered layer $0$), which corresponds to Macbeath regions of volume $v = \eps^{(d+1)/2}$, has the maximum thickness $\eps$, and a layer of number $i$ (which may be positive or negative) has thickness of just $\eps / i^2$. As a result, the sum of the thicknesses of all layers is then given by $\sum_i \eps/i^2 = O(\eps)$, which eliminates the wasteful log factor. 

This causes a problem however. Macbeath regions inside each layer have widths proportional to the thickness of the layer, and a naive analysis would suggest that more Macbeath regions would be needed as the absolute value of $i$ increases. The resulting increase in the number of Macbeath regions would negate the savings achieved in the total thickness. The volume-sensitive bound saves us. Omitting technical details, we show that as their thicknesses decrease, the volumes change in a compensatory manner, and the volume-sensitive bound implies that the feared blowup in numbers does not occur. As a consequence, the total number of Macbeath regions can be shown to be just $O(1/\eps^{(d-1)/2})$, yielding the optimal bound for the combinatorial complexity.

The remainder of the paper is organized as follows. In Section~\ref{s:prelim}, we introduce the various geometric preliminaries upon which our construction relies, and summarize the salient properties of Macbeath regions, which are central to our construction. In Section~\ref{s:polar-caps}, we investigate the relationship between $\eps$-width caps in the primal body $K$ and its polar $K^*$. In Section~\ref{s:volume-sensitive-bound}, we show that the number of disjoint Macbeath regions of width $\eps$ and volume $\Theta(v)$ is $O(\min(\eps/v, v/\eps^d))$. Finally, we prove Theorem~\ref{thm:main} in Section~\ref{s:opt}. Conclusions and open problems are discussed in Section~\ref{s:conclusion}.

\section{Geometric Preliminaries} \label{s:prelim}

In this section we present a number of concepts and observations that will be used throughout the paper. Much of the material in this section has been presented in \cite{AFM17a, AFM17b, AFM17c}. We have included it here for the sake of completeness. The proofs of all the lemmas in this subsection that are omitted can be found in these papers or are straightforward adaptations of the proofs given therein.

Consider a convex body $K$ in $d$-dimensional space $\RE^d$. Let $\bd K$ denote the boundary of $K$. Let $O$ denote the origin of $\RE^d$. Given a parameter $0 < \gamma \le 1$, we say that $K$ is \emph{$\gamma$-fat} if there exist concentric Euclidean balls $B$ and $B'$, such that $B \subseteq K \subseteq B'$, and $\radius(B) / \radius(B') \ge \gamma$. We say that $K$ is \emph{fat} if it is $\gamma$-fat for a constant $\gamma$ (possibly depending on $d$, but not on $\eps$). We will use $\vol(K)$ to denote its $d$-dimensional Lebesque measure. When dealing with $(d-1)$-dimensional convex objects, we will use $\area(\cdot)$ to denote the $(d-1)$-dimensional Lebesque measure.

Let $B_0$ denote the ball of unit radius centered at the origin and for $\alpha > 0$, let $\alpha B_0$ denote the ball of radius $\alpha$ centered at the origin. We say that a convex body $K$ is in \emph{$\gamma$-canonical form}%
\footnote{This definition differs from our earlier papers but has the elegant feature that a convex body $K$ is in $\gamma$-canonical form if and only if its polar $K^*$ (see Section~\ref{s:prelim-polar}) is as well.}
if it is nested between $\sqrt{\gamma} B_0$ and $B_0 / \sqrt{\gamma}$ (see Figure~\ref{f:prelim}(a)). A body in $\gamma$-canonical form is $\gamma$-fat and, for constant $\gamma$, it is fat and has $\Theta(1)$ diameter.

\begin{figure*}[btp]
  \centerline{\includegraphics[scale=.65]{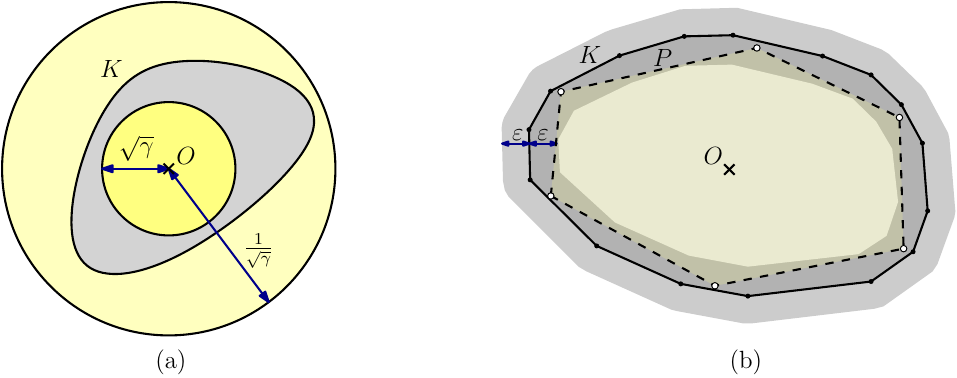}}
  \caption{\label{f:prelim}(a) A convex body $K$ in $\gamma$-canonical form and (b) an inner $\eps$-approximation $P$.}
\end{figure*}

We say that a convex body $P$ is an \emph{$\eps$-approximation} (see Figure~\ref{f:prelim}(b)) to another convex body $K$ if they are within Hausdorff error $\eps$ of each other. Further, we say that $P$ is an \emph{inner} (resp., \emph{outer}) approximation if $P \subseteq K$ (resp., $P \supseteq K$).
The next lemma shows that, up to constant factors, the problem of approximating a convex body can be reduced to the problem of approximating a convex body in canonical form. The proof is an easy consequence of John's Theorem~\cite{John}. (Also, see Lemma~{2.1} of \cite{AFM17a}.)

\begin{lemma} \label{lem:canonical}
Let $K \subset \RE^d$ be a convex body of unit diameter. There exists a non-singular affine transformation $T$ such that $T(K)$ is in $(1/d)$-canonical form and if $P$ is any $(2\eps/\sqrt{d})$-approximation to $T(K)$, then $T^{-1}(P)$ is an $\eps$-approximating polytope to $K$.
\end{lemma}

In light of this result, we may assume that $K$ is presented in $\gamma$-canonical form, for any constant $\gamma$ (depending on dimension), and that $\eps$ has been appropriately scaled. (This scaling will only affect the constant factors hidden in our asymptotic bounds. The transformation also preserves directionally sensitive notions of approximation~\cite{AFM17c}.) Henceforth, we will focus on the problem of $\eps$-approximating a convex body $K$ in canonical form.

Finally, we define two useful notions of distance from the boundary of a convex body. Let $K$ be a convex body and let $x$ be a point. Define $\delta(x)$ to be the minimum distance from $x$ to any point on $\bd K$. We define a ray-based notion of distance of a point $x$ as well. Consider the intersection point $p$ of $\bd K$ and the ray emanating from $O$ and passing through $x$. Define $x$'s \emph{ray-distance}, denoted $\ray(x)$, to be $\|x p\|$ (see Figure~\ref{f:cap}(a)). We have the following lemma, which shows that for points inside a convex body in $\gamma$-canonical form for constant $\gamma$, these two distance measures are the same to within a constant factor. The lemma is a straightforward modification of Lemma~{4.2} in \cite{AFM17a} (by replacing the ratios of the radii of the inner and outer balls, which was $1/d$ there and $\gamma$ here).

\begin{lemma} \label{lem:raydist-delta}
Let $K$ be a convex body in $\gamma$-canonical form. For any point $x \in K$, $\delta(x) \le \ray(x) \le \delta(x) / \gamma$.
\end{lemma}

\subsection{Caps and Macbeath Regions} \label{s:prelim-macbeath}

Given a convex body $K$, a \emph{cap} $C$ is defined to be the nonempty intersection of the convex body $K$ with a halfspace (see Figure~\ref{f:cap}(b)). Let $h$ denote the hyperplane bounding this halfspace. We define the \emph{base} of $C$ to be $h \cap K$. The \emph{apex} of $C$ is any point in the cap such that the supporting hyperplane of $K$ at this point is parallel to $h$. The \emph{width} of $C$, denoted $\width(C)$, is the distance between $h$ and this supporting hyperplane. Given any unit vector $u$ and any sufficiently small width $w$, there is a unique cap of width $w$ whose base is orthogonal to $u$ and lies on the same side of the origin as indicated by $u$. We refer to this as the cap that is \emph{orthogonal} to $u$. Given any cap $C$ of width $w$ and a real $\rho \ge 0$, we define its \emph{$\rho$-expansion}, denoted $C^{\rho}$, to be the cap of $K$ of width $\rho w$ cut by a hyperplane parallel to the base of $C$. 
(Note that $C^{\rho} = K$, if $\rho w$ exceeds the width of $K$ along the defining direction.)

\begin{figure*}[tbp]
  \centerline{\includegraphics[scale=.65]{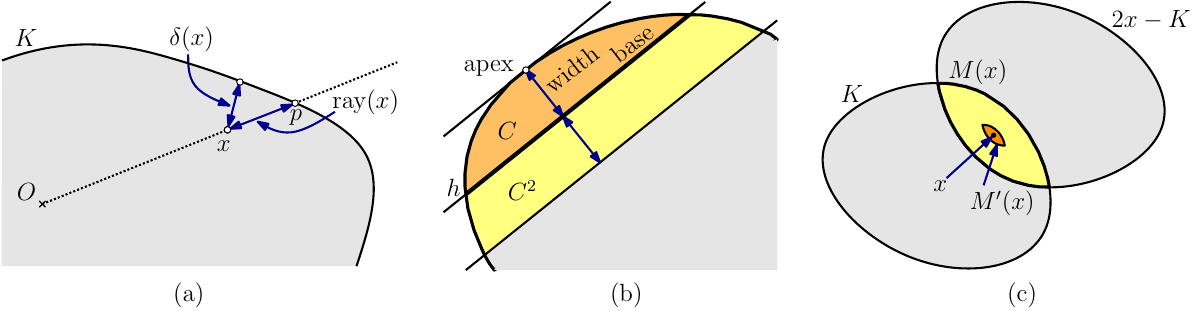}}
  \caption{\label{f:cap}(a) Notions of distance, (b) cap concepts, and (c) Macbeath regions.}
\end{figure*}

We begin with some simple geometric facts about caps.  An easy consequence of convexity is that, for $\rho \ge 1$, $C^{\rho}$ is a subset of the region obtained by scaling $C$ by a factor of $\rho$ about its apex. This implies the following lemma.

\begin{lemma} \label{lem:cap-exp}
Let $K \subset \RE^d$ be a convex body and $\rho \ge 1$. For any cap $C$ of $K$, $\vol(C^{\rho}) \le \rho^d \cdot \vol(C)$.
\end{lemma}

Another consequence of convexity is that containment of caps is preserved if the halfspaces defining both caps are consistently scaled about a point that is common to both caps. This is stated in the following lemma, which appeared as Lemma~{4.4} in~\cite{AFM17a}.

\begin{lemma} \label{lem:cap-containment}
Let $K$ be a convex body, and let $C_1 \subseteq C_2$ be two caps of $K$. Let $H_1$ and $H_2$ be their respective defining halfspaces, and let $H^{\lambda}_1$ and $H^{\lambda}_2$ be the respective halfspaces obtained by scaling by $\lambda \geq 1$ about any point $p \in C_1$. Then $K \cap H^{\lambda}_1 \,\subseteq\, K \cap H^{\lambda}_2$.
\end{lemma}

The next two lemmas apply to bodies in $\gamma$-canonical form. The first shows that for a point $p$ in a cap near the boundary, the angle between $O p$ and the normal to the base of the cap is bounded away from $\pi/2$. The second gives upper and lower bounds on the volume of a cap of width $\alpha$. These two lemmas are straightforward adaptations of Lemma 2.11 in \cite{AFM17b} and Lemma 4.3 in \cite{AFM17a}, respectively.

\begin{lemma} \label{lem:raydist}
Let $K \subset \RE^d$ be a convex body in $\gamma$-canonical form for constant $\gamma$, and let $\Delta_0$ be a sufficiently small constant (depending on $d$ and $\gamma$). Let $C$ be a cap of width at most $\Delta_0$ and let $p$ be any point inside $C$. Then the cosine of the angle between $O p$ and the normal to the base of $C$ is at least $\gamma/2$.
\end{lemma}

\begin{lemma} \label{lem:width-vol}
Let $K \subset \RE^d$ be a convex body in $\gamma$-canonical form for constant $\gamma$, and let $0 < \alpha < 1$ be a positive real. Then there exist constants $c$ and $c'$ (depending on $d$ and $\gamma$) such that for any cap $C$ of width $\alpha$, $c \kern+1pt \alpha^d \leq \vol(C) \leq c' \alpha$.
\end{lemma}

Given a point $x \in K$ and a real parameter $\lambda \ge 0$, the \emph{Macbeath region} $M^{\lambda}(x)$ (also called an \emph{M-region}) is defined as:
\[
  M^{\lambda}(x) ~ = ~ x + \lambda ((K-x) \cap (x-K)).
\]
It is easy to see that $M^{1}(x)$ is the intersection of $K$ and the reflection of $K$ around $x$ (see Figure~\ref{f:cap}(c)), and so $M^{1}(x)$ is centrally symmetric about $x$. $M^{\lambda}(x)$ is a scaled copy of $M^{1}(x)$ by the factor $\lambda$ about $x$. We refer to $x$ as the \emph{center} of $M^{\lambda}(x)$ and to $\lambda$ as its \emph{scaling factor}. As a convenience, we define $M(x) = M^1(x)$ and $M'(x) = M^{1/5}(x)$. To emphasize the difference from unit-scale Macbeath regions, we often refer to the latter as \emph{shrunken} Macbeath regions.

Macbeath regions have found numerous uses in the theory of convex sets and the geometry of numbers (see B\'{a}r\'{a}ny~\cite{Bar00} for an excellent survey). They have also been applied to a growing number of results in the field of computational geometry, particularly to construct lower bounds for range searching~\cite{BCP93, AMX12, AMM09b} and to bound the complexity of an $\eps$-approximating polytope~\cite{AFM12b,AFM17a}.

Given any point $x \in K$, let $C(x)$ denote the cap of minimum volume that contains $x$. There generally may be multiple minimum-volume caps containing a given point, and if so, $C(x)$ denotes any such cap. Clearly, the base of $C(x)$ must pass through $x$. In fact, a standard variational argument implies $x$ is the centroid of the base (for otherwise, we could decrease the cap volume by an infinitesimal rotation of the base about $x$~\cite{ELR70}). Indeed, our use of minimum-volume caps is primarily due to the fact that they are well centered around their defining point, and not volume properties.

We now present lemmas that encapsulate key properties of Macbeath regions, which will be useful in our analysis. The first lemma shows that if two shrunken Macbeath regions have a nonempty intersection, then a constant factor expansion of one contains the other~\cite{BCP93,ELR70}. Our version uses different parameters and is proved in \cite{AFM17b} (Lemma 2.4).

\begin{lemma}\label{lem:mac-mac}
Let $K$ be a convex body, and let $\lambda \le 1/5$ be any real. If $x,y \in K$ such that $M^{\lambda}(x) \cap M^{\lambda}(y) \neq \emptyset$, then $M^{\lambda}(y) \subseteq M^{4\lambda}(x)$. \end{lemma}

The next two lemmas are useful in situations when we know that a Macbeath region partially overlaps a cap of $K$, and allow us to conclude that a constant factor expansion of the cap will fully contain the Macbeath region. The first applies to shrunken Macbeath regions and the second to Macbeath regions with scaling factor one. The proof of the first appears in \cite{AFM17a} (Lemma 2.5), and the second is an immediate consequence of the definition of Macbeath regions.

\begin{lemma} \label{lem:mac-cap2}
Let $K$ be a convex body. Let $C$ be a cap of $K$ and $x$ be a point in $K$ such that $C \cap M'(x) \neq \emptyset$. Then $M'(x) \subseteq C^2$.
\end{lemma}

\begin{lemma} \label{lem:mac-cap2-var}
Let $K$ be a convex body. If $x$ is a point in a cap $C$ of $K$, then $M(x) \subseteq C^2$.
\end{lemma}

The next lemma shows that any sufficiently small cap is contained within a suitable constant factor expansion of the Macbeath region centered at the centroid of its base. In particular, it implies that the minimum-volume cap associated with a point is contained within a suitable constant factor expansion of the Macbeath region centered at that point. The proof is a straightforward adaptation of proofs in~\cite{BCP93} and~\cite{ELR70}.

\begin{lemma} \label{lem:mac2-var}
Let $K \subset \RE^d$ be a convex body in $\gamma$-canonical form for constant $\gamma$, and let $\Delta_0$ be a sufficiently small constant (depending on $d$ and $\gamma$). Let $C$ be a cap of $K$ of width at most $\Delta_0$ and let $x$ denote the centroid of the base of this cap. Then $C \subseteq M^{3d}(x)$. 
\end{lemma}

The following four lemmas are easy consequences of standard properties of Macbeath regions. The first gives lower and upper bounds on the width of the minimum-volume cap. Recalling that $x$ is the centroid of the base of $C(x)$, the second generalizes Lemma~\ref{lem:mac2-var} to expansions of minimum-volume caps. The third states that if the shrunken Macbeath regions associated with two caps overlap, then a constant factor expansion of any one cap is contained in a suitable constant factor expansion of the other. The last states that all the points in a shrunken Macbeath region have similar distances from the boundary of $K$. These lemmas are straightforward adaptations of Lemma~2.10 in \cite{AFM17b}, Lemmas~2.8 and~2.9 in \cite{AFM17a}, and Lemma~2.14 in \cite{AFM17b}, respectively.

\begin{lemma} \label{lem:width-delta}
Let $K, C$ and $x$ be as defined in Lemma~\ref{lem:mac2-var}. Then $\delta(x) \le \width(C) \le c \kern+1pt \delta(x)$ for a suitable constant $c$ (depending on $d$ and $\gamma$).
\end{lemma}

\begin{lemma} \label{lem:mac2-var-exp}
Let $\rho \ge 1$ and let $K, C$, and $x$ be as defined in Lemma~\ref{lem:mac2-var}. Then  $C^{\rho} \subseteq M^{3d(2\rho-1)}(x)$.
\end{lemma}

\begin{lemma} \label{lem:mac-mac-var}
Let $K \subset \RE^d$ be a convex body in $\gamma$-canonical form, where $\gamma$ is a constant. Let $\Delta_0$ be the constant of Lemma~\ref{lem:mac2-var} and let $\rho \ge 1$ be any real. There exists a constant $\beta \ge 2$ such that the following holds. Let $C_1$ and $C_2$ be any two caps of $K$ of width at most $\Delta_0$. Let $x_1$ and $x_2$ denote the centroids of the bases of the caps $C_1$ and $C_2$, respectively. If $M'(x_1) \cap M'(x_2) \neq \emptyset$, then $C_1^{\rho} \subseteq C_2^{\beta \rho}$.
\end{lemma}

\begin{lemma} \label{lem:core-delta}
Let $K$ be a convex body. If $x \in K$ and $x' \in M'(x)$, then $4 \delta(x) / 5 \le \delta(x') \le 4 \delta(x) / 3$.
\end{lemma}

\subsection{Polarity and the Mahler Volume} \label{s:prelim-polar}

Some of our analysis will involve the well known concept of \emph{polarity}. Let us recall some general facts (see, e.g., Eggleston~\cite{Egg58}). Given vectors $u, v \in \RE^d$, let $\ang{u,v}$ denote their dot product, and let $\|v\| = \sqrt{\ang{v,v}}$ denote $v$'s Euclidean length. (Throughout, we use the terms \emph{point} and \emph{vector} interchangeably.) Given a bounded convex body $K \in \RE^d$ that contains the origin in its interior, define its \emph{polar}, denoted $K^*$, to be the convex set
\[
	K^*
		~ = ~ \{ u : \ang{u,v} \le 1, \hbox{~for all $v \in K$} \}.
\]
The polar enjoys many useful properties. For example, it is well known that $K^*$ is bounded and $(K^*)^* = K$. Further, if $K_1$ and $K_2$ are two convex bodies such that $K_1 \subseteq K_2$ then $K_2^* \subseteq K_1^*$. 

It will be convenient to also define the \emph{polar} of a point. Given a point $v \in \RE^d$ (not the origin), we define $v^*$ to be the hyperplane that is orthogonal to $v$ and at distance $1/\|v\|$ from the origin, on the same side of the origin as $v$. The polar of a hyperplane that does not contain the origin is defined as the inverse of this mapping. We may equivalently define $K^*$ as the intersection of the closed halfspaces that contain the origin, bounded by the hyperplanes $v^*$, for all $v \in K$. 

For any centrally symmetric convex body $K$ and real $\lambda > 0$, we use $\lambda K$ to denote the body obtained by scaling $K$ by a factor of $\lambda$ about its center. Given any convex body $K$ and real $\lambda \geq 1$,  we say that a point $x \in K$ is $\lambda$-centered in $K$ if there exists an ellipsoid $\Phi$ centered at $x$ such that $\Phi \subseteq K \subseteq \lambda \Phi$. It is well known that every convex body $K$ contains a $d$-centered point. This follows from the fact that the ellipsoid $\Phi$ of largest volume contained inside a convex body $K$, called the \emph{John Ellipsoid}, satisfies $\Phi \subseteq K \subseteq d \Phi$. The following lemma shows that the centroid of $K$ is $O(1)$-centered (the proof is presented in~\cite{KLS95}).

\begin{lemma} \label{lem:john-centroid}
For any convex body $K$, its centroid is $d$-centered in $K$.
\end{lemma}

An important concept related to polarity is the \emph{Mahler volume}, which is defined to be the product of the volumes of a convex body and its polar. There is a large literature on the Mahler volume (see, e.g., Kuperberg~\cite{Kuperberg}). It is well known that the Mahler volume is affine invariant. The following lemma establishes a constant upper and lower bound on the Mahler volume of a convex body when the polar is computed with respect to an $O(1)$-centered point. Note that the convex body need not be centrally symmetric.

\begin{lemma} \label{lem:mahler}
Let $K$ be a convex body, let $x \in K$ be a $\lambda$-centered point for constant $\lambda$, and let $K^*$ denote the polar of $K$ with respect to $x$. Then $\vol(K) \cdot \vol(K^*) = \Theta(1)$.
\end{lemma}

\begin{proof}
By definition, there is an ellipsoid $\Phi$ centered at $x$ such that $\Phi \subseteq K \subseteq \lambda \Phi$. By standard properties of the polar transformation, we have $(\lambda \Phi)^* \subseteq K^* \subseteq \Phi^*$. Since $(\lambda \Phi)^* = (1/\lambda) \Phi^*$, it follows that $(1/\lambda) \Phi^* \subseteq K^* \subseteq \Phi^*$. Thus 
 \[
     \frac{1}{\lambda^d} \cdot \vol(\Phi) \cdot \vol(\Phi^*) 
        ~ \leq ~ \vol(K) \cdot \vol(K^*)
        ~ \leq ~ \lambda^d \cdot \vol(\Phi) \cdot \vol(\Phi^*).
 \]
Since the Mahler volume is invariant under linear invertible transformation, $\vol(\Phi) \cdot \vol(\Phi^*)$ equals the square of the volume of a unit ball, which is $\Theta(1)$. The desired result follows immediately.
\end{proof}

The following lemma is now an immediate consequence of Lemmas~\ref{lem:john-centroid} and \ref{lem:mahler}.

\begin{lemma} \label{lem:centroid-mahler}
Let $K$ be a convex body and let $K^*$ denote the polar of $K$ with respect to its centroid. Then $\vol(K) \cdot \vol(K^*) = \Theta(1)$.
\end{lemma}

\section{Caps of the Polar}
\label{s:polar-caps}

Let $K$ denote a convex body. The goal of this section is to establish certain fundamental and novel relationships between $\eps$-width caps in $K$ and its polar $K^*$. Assuming $K$ satisfies certain fatness assumptions, we show that the base of a cap in $K$ and the base of a corresponding cap in $K^*$, treated as $(d-1)$-dimensional bodies, are roughly related as polars of each other. We establish this in Lemmas~\ref{lem:polars} and \ref{lem:dualcaps}.

It will be useful to consider the notion of a cap in a dual setting (see, e.g., \cite{AFM12a,AFM12b}). Given a convex body $K$ and a point $z$ that is exterior to $K$, we define the \emph{dual cap} of $K$ with respect to $z$ to be the set of $(d-1)$-dimensional hyperplanes that pass through $z$ and do not intersect $K$'s interior (see Figure~\ref{f:dual-cap-def}). We can define the polar of a dual cap as the set of points that results by taking the polar of each hyperplane of the dual cap. If $K$ is full dimensional and contains the origin, it follows that a hyperplane $h$ lies in the dual cap if and only if the point $h^*$ lies on the base of the cap of $K^*$ defined by the hyperplane $z^*$.

\begin{figure}[htbp]
  \centerline{\includegraphics[scale=.65]{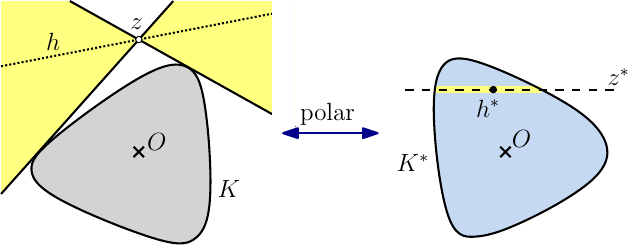}}
  \caption{\label{f:dual-cap-def}Definition of a dual cap and its polar.}
\end{figure}

The following lemma will be useful in proving Lemma~\ref{lem:dualcaps}, wherein we establish a polar relationship between the base of a cap in a convex body and the base of a corresponding cap in its polar. For the sake of concreteness, we state our results in terms of an arbitrary direction, which we call \emph{vertical}, and any hyperplane orthogonal to this direction is called \emph{horizontal}. Since the direction is arbitrary, there is no loss of generality.

\newcommand{\bBarStar}{\overline{B}^{\kern+1pt *}}

\begin{lemma} \label{lem:polars}
Let $z \in \RE^d$ be a point that lies on a vertical ray from the origin $O$, and let $B$ be a $(d-1)$-dimensional convex body whose interior intersects the segment $O z$ at some point $x$ (see Figure~\ref{f:polars}). Let $G$ be the polar of the dual cap of $B$ with respect to $z$, let $\overline{B}$ be the vertical projection of $B$, and let $h$ be the hyperplane parallel to $B$ passing through $z$. Then $G - h^* = \alpha \bBarStar$, where $\alpha = \|x z\| / \|O z\|$.
\end{lemma}

\begin{figure}[htbp]
  \centerline{\includegraphics[scale=.65]{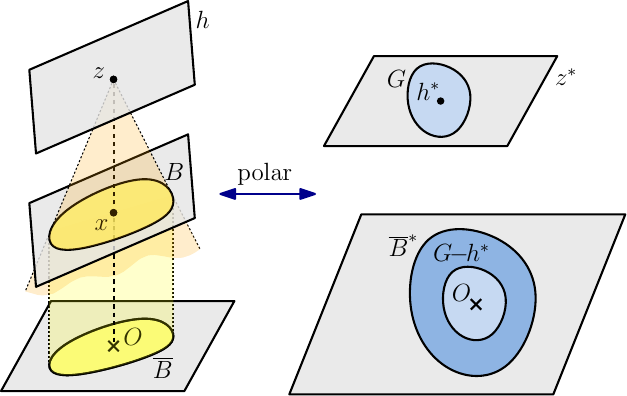}}
  \caption{\label{f:polars}Statement of Lemma~\ref{lem:polars}.}
\end{figure}

Note that $G$ is a $(d-1)$-dimensional convex body that lies on the horizontal hyperplane $z^*$. Since $h$ passes through $z$ and does not intersect $B$, $h^*$ is a point lying within $G$, and therefore the translate $G - h^*$ is horizontal and contains the origin. Since $\overline{B}$ is also horizontal and contains the origin, so does $\bBarStar$ (see Figure~\ref{f:polars}).

\begin{proof}
Let $g$ be any hyperplane passing through $z$ and not intersecting $B$. Note that $g$ is the set of points $p$ satisfying $\ang{g^*, p} = 1$. Define the hyperplane $g'$ to be the set of points $p$ satisfying the equation $\ang{g^*-h^*,p} = \alpha$. We claim that (i) the normal to $g'$ has no vertical component, and (ii) it passes through the intersection of $g$ with the hyperplane containing $B$. 

Note that (i) follows from the fact that $\ang{g^*-h^*,z} = \ang{g^*,z} - \ang{h^*,z} = 0$, since $\ang{g^*,z} = \ang{h^*,z} = 1$. To prove (ii),  let $y$ be any point on the intersection of $g$ with the hyperplane containing $B$. We have
\[
     \ang{g^* - h^*, y} 
        ~ = ~ \ang{g^*,y} \kern-1pt - \kern-1pt \ang{h^*,y} \kern-1pt 
        ~ = 1 - \ang{h^*,x},
\]
since $\ang{g^*,y} = 1$, and $\ang{h^*,y} = \ang{h^*,x}$. Since $\ang{h^*,z} = 1$, we have $\ang{h^*,x} = \|x\| / \|z\| = 1 - \alpha$. Thus $\ang{g^* - h^*, y} = \alpha$, which proves (ii).

Clearly the polar of $\overline{B}$ is the set of all points $(g')^*$ that can be generated in this manner. We have $(g')^* = (1/\alpha) (g^*-h^*)$. It follows that $G - h^* = \alpha\bBarStar$, as desired.
\end{proof}

For the rest of this section, we focus on the case where $K$ is a convex body in $\gamma$-canonical form for constant $\gamma$. Before presenting our next lemma, we introduce a mapping of $\eps$-width caps of $K$ to $\Theta(\eps)$-width caps of $K^*$. Let $C$ be an $\eps$-width cap of $K$. For concreteness, assume that space has been rotated so that $C$'s base is horizontal with $C$ lying above the origin. For a suitably large constant $c$ (to be specified in the proof of the following lemma), shoot a ray vertically upwards from $O$, and let $x \in K^*$ be a point on this ray such that $\delta(x) = \eps/c$ (where $\delta(x)$ is defined with respect to $K^*$). Define $\pi(C)$ to be a minimum-volume cap of $K^*$ that contains $x$ (see Figure~\ref{f:dualcaps}). Our next lemma shows that, up to constant factors, there is a polar relationship between the bases of $C$ and $\pi(C)$.

\begin{figure*}[htbp]
  \centering
  \includegraphics[scale=.65]{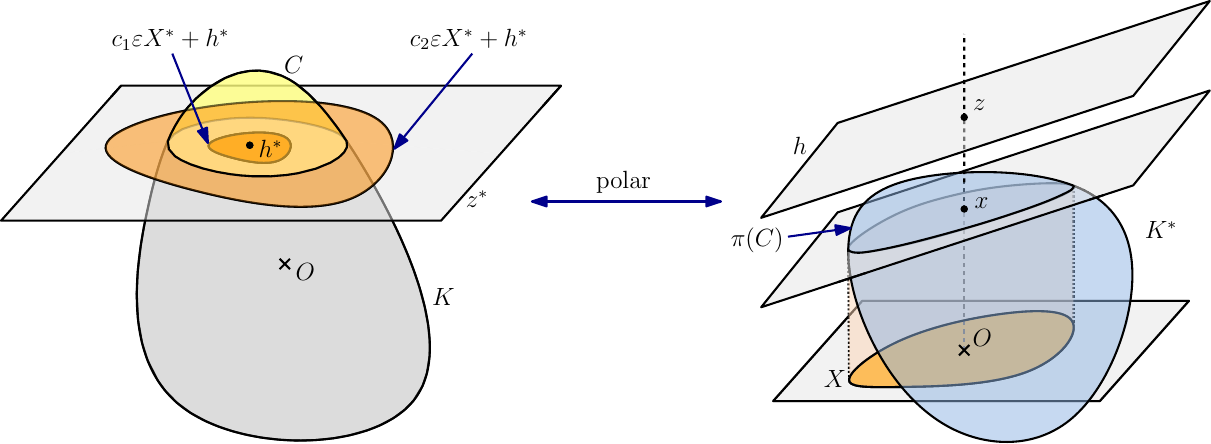}
  \caption{\label{f:dualcaps}Statement of Lemma~\ref{lem:dualcaps}.}
\end{figure*}

\begin{lemma} \label{lem:dualcaps}
Given a convex body $K$ in $\gamma$-canonical form for a constant $\gamma$, there exist constants $c_1, c_2$ such that the following holds. Let $C$ be a horizontal $\eps$-width cap of $K$ lying above the origin. Let $X$ be the vertical projection of the base of $\pi(C)$. Let $z$ be the point that is the polar of the hyperplane passing through $\base(C)$, and let $h$ be the hyperplane parallel to $\base(\pi(C))$ passing through $z$. Then $c_1 \eps X^* \subseteq \base(C) - h^* \subseteq c_2 \eps X^*$.
\end{lemma}

\begin{proof}
Let $z'$ be the point of intersection of the vertical ray $O x$ with $\bd K^*$. Recall that $K$ is in $\gamma$-canonical form for constant $\gamma$, $\width(C) = \eps$, $z$ is the polar of the hyperplane passing through $\base(C)$, and $x \in K^*$ is a point on the ray $O z$ such that $\delta(x) = \eps/c$. It is now easy to show from the definition of the polar transformation that $z$ lies on the ray $O x$ (outside $K^*$) such that $c'_1 \eps \le \|z' z\|  \le c'_2 \eps$, for suitable constants $c'_1$ and $c'_2$ depending on $\gamma$.
Let $B_1 = \base(\pi(C))$. Let $D$ and $D_1$ be the dual caps of $K^*$ and $B_1$, respectively, with respect to $z$. Recall that the polar of $D$ is $\base(C)$. Let $G_1$ be the polar of $D_1$. Since $B_1 \subseteq K^*$, we have $D \subseteq D_1$. It follows that $\base(C) \subseteq G_1$. Applying Lemma~\ref{lem:polars}, we obtain $G_1 - h^* = (\|x z\|/\|O z\|) X^*$. Since $\|x z\| = \Theta(\eps)$ and $\|O z\| = \Theta(1)$, it follows that $\base(C) - h^* \subseteq c_2 \eps X^*$, for a suitable constant $c_2$. This establishes one part of the desired inequality.

We establish the other part of the inequality using a similar approach. The key insight is that a scaled and translated version of $B_1$ has the property that its dual cap with respect to $z$ is a subset of $D$, the dual cap of $K^*$ with respect to $z$. Towards this end, let $\Psi_1$ be the (infinite) generalized cone defined by rays shot from $z'$ through $B_1$ (see Figure~\ref{f:dualcaps2}). By convexity, $\Psi_1$ encloses $K^* \setminus \pi(C)$. Consider the generalized (finite) cone $\Gamma_1$ with apex $O$ and base $B_1$. Scale this cone about $O$ by a suitable factor $f_1$ to produce a cone $\Gamma_2$ whose base (call it $B_2$) is tangent to $K^*$. Since $\delta(x) = \eps/c$ and, by Lemma~\ref{lem:width-delta}, $\width(\pi(C)) = O(\delta(x))$, it is easy to see that $f_1 \le 1 + c'_3(\eps/c)$ for a suitable constant $c'_3$. By convexity, $\Gamma_2$ encloses the entirety of $\pi(C)$. 

\begin{figure*}[htbp]
  \centering
  \includegraphics[scale=.65]{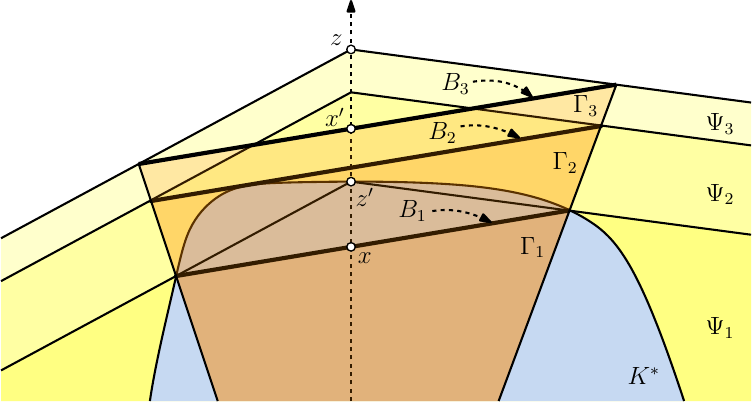}
  \caption{\label{f:dualcaps2}Definitions in the proof of Lemma~\ref{lem:dualcaps} (origin not shown).}
\end{figure*}

We now expand $\Psi_1$ by this same factor $f_1$ about $O$ to produce $\Psi_2$. Clearly, $\Psi_2$ encloses both $K^* \setminus \pi(C)$ and $\pi(C)$ itself, so it encloses all of $K^*$. For sufficiently large constant $c$, we can show that the apex of $\Psi_2$ lies below $z$. This is because the distance between $z'$ and the apex of $\Psi_2$ is $(f_1-1) \|O z'\| \le c'_3(\eps/c) \|O z'\| \le c'_3 (\eps/c) (1/\sqrt{\gamma})$. Recall that $z$ is at distance at least $c'_1 \eps$ from $z'$. Thus, by choosing $c$ larger than $c'_3/(c'_1 \sqrt{\gamma})$, we can ensure that the apex of $\Psi_2$ lies below $z$. We scale $\Psi_1$ about $O$ such that the apex of the resulting cone $\Psi_3$ is $z$. The scaling factor $f_2$ is $\|O z\| / \|O z'\| = 1+\Theta(\eps)$. We scale $\Psi_1$ and $B_1$ by the same factor $f_2$ about $O$ to obtain $\Psi_3$ and $B_3$, respectively. Clearly $\Psi_3$ also encloses the entirety of $K^*$. Further, $B_3$ is at vertical distance $\Theta(\eps)$ below $z$. To see this, let $x'$ be the point of intersection of $O z$ with $B_3$. Note that $\|x' z\| = f_2 \|x z'\| = \Theta(\eps)$. 

Let $D_3$ be the dual cap of $B_3$ with respect to $z$, and let $G_3$ be the polar of $D_3$. Since $D_3 \subseteq D$, it follows that $G_3 \subseteq \base(C)$. Applying Lemma~\ref{lem:polars}, we obtain $G_3 - h^* = (\|x' z\| / \|O z\|)(X')^*$, where $X'$ is the vertical projection of $B_3$ (and so is a constant-factor expansion of $X$). Since $\|x' z\| = \Theta(\eps)$ and $\|O z\| = \Theta(1)$, it follows that $G_3 - h^* = c_1 \eps X^*$, for some constant $c_1$. Thus, we have $c_1 \eps X^* \subseteq \base(C) - h^*$, which is the other part of our desired inequality.
\end{proof}

Given this polar-like relationship between $C$ and $\pi(C)$, we can apply the Mahler volume to bound the product of their volumes.

\begin{lemma} \label{lem:mahler-caps}
Let $C$ be as defined in Lemma~\ref{lem:dualcaps}. Then $\vol(C) \cdot \vol(\pi(C)) = \Theta(\eps^{d+1})$.
\end{lemma}

\begin{proof}
Recall that $K$ is in $\gamma$-canonical form for constant $\gamma$. Thus $K^*$ is also in $\gamma$-canonical form for constant $\gamma$ and, by Lemma~\ref{lem:raydist}, $\area(\base(\pi(C))) = \Theta(\area(X))$. By Lemma~\ref{lem:dualcaps}, $\base(C)$ is sandwiched between two scaled copies of $X^*$, where the scaling factor is $\Theta(\eps)$. Thus $\area(\base(C)) = \Theta(\eps^{d-1} \area(X^*))$. We have 
\[ 
    \area(\base(C)) \cdot \area(\base(\pi(C))) 
        ~ = ~ \Theta(\eps^{d-1} \area(X^*) \cdot \area(X)) = \Theta(\eps^{d-1}). 
\]
In the last step we have used the fact that for any minimum-volume cap containing a point $x$, $x$ is the centroid of the base of the cap. Thus, it follows from the definition of $\pi(C)$ and $X$ that $X^*$ is the polar of $X$ with the centroid of $X$ as origin, and so $\area(X^*) \cdot \area(X) = \Theta(1)$ by Lemma~\ref{lem:centroid-mahler}. Finally, the lemma follows by noting that the caps $C$ and $\pi(C)$ each have $\Theta(\eps)$ width and thus their volumes are $\Theta(\eps)$ times the areas of their respective bases.
\end{proof}

Next, we extend the above product bound on the volumes of caps $C$ and $\pi(C)$ in two respects. First, the bound applies to the product of the volumes of $C$ and a shrunken Macbeath region $R$, where $R$ serves as a proxy for the $\pi(C)$. Second, unlike $\pi(C)$, which is uniquely determined from $C$, there is considerable flexibility in the choice of $R$.

\begin{lemma} \label{lem:mahler-mac}
Given a convex body $K$ in $\gamma$-canonical form for a constant $\gamma$, let $C$ be a horizontal $\eps$-width cap of $K$ lying above the origin. For a sufficiently large constant $c$, let $y \in K^*$ be a point at distance $\eps/c$ from $\bd K^*$ such that the ray from the origin shot vertically upwards intersects the Macbeath region $R = M'(y)$ of $K^*$. Then $\vol(C) \cdot \vol(R) = \Theta(\eps^{d+1})$.
\end{lemma}

\begin{proof}
Let $x$ be a point in the intersection of $M'(y)$ with the ray shot vertically upwards from $O$, and let $E$ denote the minimum-volume cap containing $x$. By Lemma~\ref{lem:core-delta}, $4 \delta(y) / 5 \le \delta(x) \le 4 \delta(y) / 3$. It follows from Lemma~\ref{lem:mahler-caps} that $\vol(C) \cdot \vol(E) = \Theta(\eps^{d+1})$. By Lemmas~\ref{lem:mac-cap2} and \ref{lem:mac2-var}, we have $\vol(M'(x)) = \Theta(\vol(E))$. As $M'(x)$ intersects $R$, by Lemma~\ref{lem:mac-mac}, a constant factor expansion of $R$ encloses $M'(x)$ and vice versa, so $\vol(R) = \Theta(\vol(M'(x)))$. Thus $\vol(C) \cdot \vol(R) = \Theta(\eps^{d+1})$, as desired.
\end{proof}

Intuitively, the points inside a scaled Macbeath region $M'(x)$ are similar in many respects. Two points within a shrunken Macbeath region in the polar body can be thought of as generating similar caps in the original body in the sense of satisfying a ``sandwiching'' property. This property will be used in Section~\ref{s:volume-sensitive-bound} to allow us to focus on a discrete set of caps as defined by a finite collection of Macbeath regions. This will be proved in Lemma~\ref{lem:sandwich} below. The following technical lemma will be useful for proving it.

\begin{lemma} \label{lem:guarding}
Given a convex body $K$ in $\gamma$-canonical form for a constant $\gamma$, consider a Macbeath region $R = M'(y)$ for $K$ and two rays $r$ and $r'$ shot from the origin through $R$ (see Figure \ref{f:guarding}(a)). Let $z \not\in K$ be a point on $r$. Let $h$ be a hyperplane passing through $z$ that does not intersect $K$. If $\ray(z)$ and $\delta(y)$ are sufficiently small (as a function of $\gamma$), then $r'$ intersects $h$, and letting $z' = r' \cap h$, we have $\ray(z') = O(\ray(z) + \delta(y))$.
\end{lemma}

\begin{figure}[htbp]
  \centering
  \includegraphics[scale=.65,page=1]{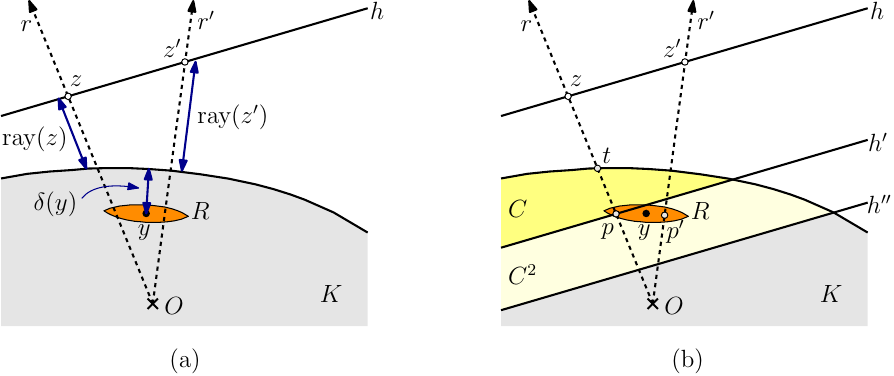}
  \caption{\label{f:guarding} (a)~Statement and (b)~proof of Lemma~\ref{lem:guarding}.}
\end{figure}

\begin{proof}
Consider a hyperplane $h'$ that is parallel to $h$ and passes through a point $p \in r \cap R$ (see Figure \ref{f:guarding}(b)). We claim that the distance between $h$ and $h'$ is $O(\ray(z) + \delta(y))$. To see this, let $t$ be the point of intersection of $r$ with $\bd K$. By Lemma~\ref{lem:core-delta}, $\delta(p) = O(\delta(y))$. Applying Lemma~\ref{lem:raydist-delta}, we have $\|p t\| = \ray(p) = O(\delta(p))$. It follows that $\|p z\| = \|p t\| + \|t z\| = \ray(p) + \ray(z) = O(\ray(z) + \delta(y))$, which proves the claim.

Let $C$ be the cap induced by $h'$. Since $C$ intersects $R$, by Lemma~\ref{lem:mac-cap2} the cap $C^2$ encloses $R$. Let $h''$ denote the hyperplane passing through the base of $C^2$. Observe that the distance between $h''$ and $h$ is no more than twice the distance between $h'$ and $h$, and is thus $O(\ray(z) + \delta(y))$. 

By our hypothesis that $\ray(z)$ and $\delta(y)$ are sufficiently small, we may assume that $\width(C^2)$ is similarly small (specifically at most $\Delta_0$ of Lemma~\ref{lem:raydist}), and hence the ray $O p'$ passes through the base of $C^2$. This implies that this ray intersects $h$. Let $p'$ be any point in $r' \cap R$. It follows easily from Lemma~\ref{lem:raydist} that $\|p'z'\|$ is at most a constant times the distance between $h''$ and $h$. Thus, $\|p'z'\| = O(\ray(z) + \delta(y))$. It follows that $\ray(z') = O(\ray(z) + \delta(y))$, as desired.
\end{proof}

In order to establish a relationship between points in $K$ and its polar $K^*$, we show next that if two rays stab a given shrunken Macbeath region near $K^*$'s boundary, the points lying on these rays and just outside of $K^*$ induce similar caps in $K$.

\begin{lemma} \label{lem:sandwich}
Consider a convex body $K$ in $\gamma$-canonical form for a constant $\gamma$. For any choice of constants $c_1$ and $c_2$, there exist constants $c_0$ and $\sigma$, such that the following holds. For all sufficiently small $\eps$ (depending on $\gamma$, $c_1$, and $c_2$) and any $y \in K^*$ such that $\delta(y) \leq \eps/c_0$, there exists a point $z$ external to $K^*$ on the ray $O y$ such that the following property holds. Let $E$ be the cap of $K$ induced by the hyperplane $z^*$, let $r$ be any ray from the origin that passes through $M'(y)$, and let $C$ be any cap of $K$ whose base is orthogonal to $r$ such that $c_1 \eps \le \width(C) \le c_2 \eps$. Then $E^{1/\sigma} \subseteq C \subseteq E$ (see Figure \ref{f:sandwich}).
\end{lemma}

\begin{figure*}[htbp]
  \centering
  \includegraphics[scale=.65]{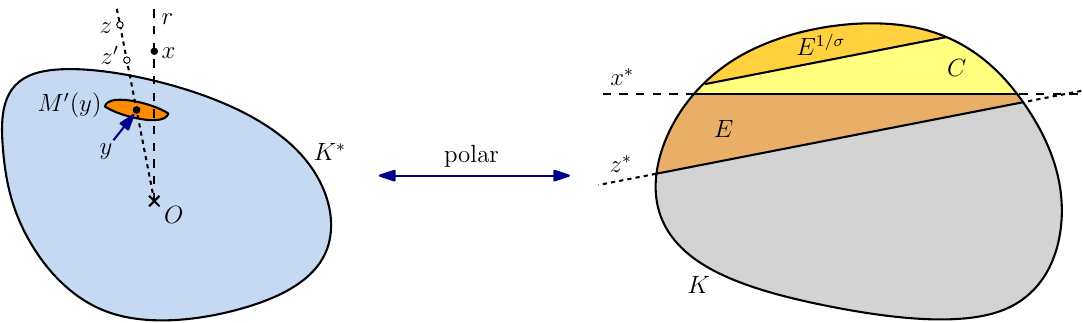}
  \caption{\label{f:sandwich}Statement and proof of Lemma~\ref{lem:sandwich}.}
\end{figure*}

\begin{proof}
Since $K$ is in $\gamma$-canonical form, it is easy to see that the polar of any hyperplane that induces a cap of $K$ of width between $c_1\eps$ and $c_2\eps$ is a point outside $K^*$, whose ray-distance with respect to $K^*$ is at least $c_1' \eps$ and at most $c_2' \eps$ for suitable constants $c_1'$ and $c_2'$ (depending only on $c_1$, $c_2$, $d$, and $\gamma$). Let $\beta$ be a sufficiently large constant. We assume that the constant $c_0$ in the statement of this lemma is $\beta/c_1'$. In other words, $\delta(y) = c_1' \eps/\beta$.

Let $z, z' \not\in K^*$ be points on the ray $O y$ such that $\ray(z) = c_2' \beta \eps$ and $\ray(z') = c_1' \eps/\beta$. Let $E$ and $E'$ be caps of $K$ induced by hyperplanes $z^*$ and $(z')^*$, respectively. Let $r$ be the ray given in the statement of the lemma (drawn vertically in Figure~\ref{f:sandwich}), and let $C$ be as described in the lemma. Let $x$ be the polar of the hyperplane passing through the base of $C$. Recall that $c_1' \eps \le \ray(x) \le c_2' \eps$. 

We claim that $E' \subseteq C \subseteq E$. Since $\ray(z') = c_1' \eps/\beta$, $\delta(y) = c_1' \eps/\beta$ and $\ray(x) \ge c_1' \eps$, it follows that $\ray(x) \ge (\beta/2)(\ray(z') + \delta(y))$. Since $\beta$ was chosen to be a sufficiently large constant, by Lemma~\ref{lem:guarding}, any hyperplane passing through $z'$ that does not intersect $K^*$ separates $K^*$ from $x$. Thus, taking the polar, it follows that $\base(E') \subseteq C$ and therefore $E' \subseteq C$. Similarly, since $\ray(x) \le c_2' \eps$, $\delta(y) = c_1' \eps/\beta$, and 
$\ray(z) = c_2' \beta \eps$, it follows that $\ray(z) \ge (\beta/2)(\ray(x) + \delta(y))$. Consequently, by Lemma~\ref{lem:guarding}, any hyperplane passing through $x$ that does not intersect $K^*$ separates $K^*$ from $z$. Thus, taking the polar, $C \subseteq E$. Putting it together, we obtain $E' \subseteq C \subseteq E$.

Let $w$ denote the distance from $O$ to $\bd K^*$ along the ray $O y$. By basic properties of the polar transformation, we have
\[
    \width(E) ~ = ~ \left(\frac{1}{w} - \frac{1}{w + c_2'\beta \eps}\right)  \qquad\hbox{and}\qquad
    \width(E') ~ = ~ \left(\frac{1}{w} - \frac{1}{w + c_1'\eps/\beta}\right).
\]
It is now easy to verify that $\width(E) / \width(E') = O(\beta^2)$, which is bounded above by some constant $\sigma$. Thus $E^{1/\sigma} \subseteq E'$. Since $E' \subseteq C \subseteq E$, it follows that $E^{1/\sigma} \subseteq C \subseteq E$. 
\end{proof}

\section{Volume of Macbeath Regions}
\label{s:volume-sensitive-bound}

In this section we present one of the main results of the paper, a volume-sensitive bound on the number of disjoint shrunken Macbeath regions associated with $\eps$-width caps. This result has the side benefit of leading to a more elegant proof of the upper bound on the total number of disjoint Macbeath regions associated with $\eps$-width caps (see, e.g, \cite{AFM17a}). Our proofs are based on the relationship we developed in the prior section between a cap in a convex body and the corresponding cap in its polar.

\begin{lemma} 
\label{lem:mac-cover}
Let $K$ be a convex body in $\gamma$-canonical form for constant $\gamma$. There exists a set $\RR$ of shrunken Macbeath regions $M'(x)$ of $K$, where $\delta(x) = \eps$, such that the following properties hold:
\begin{enumerate}
    \item[(i)] For any $v > 0$, the number of Macbeath regions in $\RR$ of volume $\Omega(v)$ is $O(\eps / v)$.
    \item[(ii)] Any ray emanating from the origin  intersects some Macbeath region of $\RR$.
\end{enumerate}  
\end{lemma}

\begin{proof}
Let $\RR'$ be a maximal set of disjoint Macbeath regions $M^{\lambda}(x)$, where $\lambda = 1/20$ and $\delta(x) = \eps$. We scale each Macbeath region of $\RR'$ by a factor of 4 about its center to obtain the set $\RR$. 

It is easy to see from Lemma~\ref{lem:core-delta} that every point in the Macbeath regions of $\RR'$ lies within distance $O(\eps)$ of $\bd K$, and so these Macbeath regions are fully contained in a region of volume $O(\eps)$. For any $v$, it follows from disjointness that $\RR'$ contains $O(\eps / v)$ Macbeath regions of volume $\Omega(v)$. Since the Macbeath regions of $\RR$ are obtained by scaling the Macbeath regions of $\RR'$ by a constant factor, this proves property (i).

To prove (ii), consider any ray emanating from the origin. Let $y$ be a point on this ray such that $\delta(y) = \eps$. Consider the Macbeath region $R_1 = M^{\lambda}(y)$. Since $\RR'$ is maximal, $R_1$ must intersect a Macbeath region $R_2 \in \RR'$. Thus, by Lemma \ref{lem:mac-mac}, the 4-factor expansion of $R_2$ fully encloses $R_1$. By construction, the 4-factor expansion of $R_2$ is an element $R \in \RR$. As $R$ encloses $R_1$, $R$ encloses $y$ as well. Thus, any ray emanating from the origin must intersect some element of $\RR$.
\end{proof}

We are now equipped to prove the main result of this section.

\begin{theorem}
\label{thm:vol-sensitive-bound}
Let $K \subset \RE^d$ be a convex body in $\gamma$-canonical form for constant $\gamma$, and let $v > 0$. Let $\CC$ be a set of caps of $K$ of width $\Theta(\eps)$ and volume $\Theta(v)$, such that the Macbeath regions $M'(x)$ centered at the centroids $x$ of the bases of these caps are disjoint. Then
\[
    |\CC| 
        ~ = ~ O\left(\min\left(\frac{\eps}{v}, \frac{v}{\eps^d} \right) \right).
\]
\end{theorem}

\begin{proof}
Let $\RR$ be the set of Macbeath regions described in the theorem. We will establish the bound given in the theorem for $\RR$. As $|\CC| = |\RR|$, this will prove the theorem. 

It is easy to see from Lemma~\ref{lem:mac-cap2} that that every point in the Macbeath regions of $\RR$ lies within distance $O(\eps)$ of $\bd K$, and so these Macbeath regions are fully contained in a region of volume $O(\eps)$. Since each such Macbeath region has volume $\Omega(v)$, it follows from disjointness that $|\RR| = O(\eps/v)$. In the rest of the proof, we will show that $|\RR| = O(v/\eps^d)$.

Construct the set of Macbeath regions as described in Lemma~\ref{lem:mac-cover}, with the convex body referred to in the lemma being the polar body $K^*$, and with $\eps$ in the lemma set to $\eps/c$ for sufficiently large constant $c$. Call this set $\RR'$. Let $c_1$ and $c_2$ be two constants such that for any cap $C \in \CC$, $c_1 \eps \le \width(C) \le c_2 \eps$. With each Macbeath region $R' \in \RR'$, we may apply Lemma~\ref{lem:sandwich} for constants $c_1$ and $c_2$ to obtain an associated cap $E$ of $K$, which we call the \emph{canonical cap} associated with $R'$. We will show that, for every Macbeath region $R \in \RR$, there is a Macbeath region $R' \in \RR'$ and an associated canonical cap $E$ satisfying the following properties: 
\begin{enumerate}
\item[(i)] $\vol(R \cap E) = \Omega(v)$,

\item[(ii)] $\vol(E) = O(v)$, and

\item[(iii)] $\vol(R') = \Omega(v')$, where $v' = \eps^{d+1} / v$.
\end{enumerate} 

Let $C$ be the cap of $\CC$ associated with $R$. Let $h$ be the hyperplane passing through the base of $C$. Shoot a ray orthogonal to $h$ from the origin. From Lemma~\ref{lem:mac-cover}(ii), we know that this ray intersects some Macbeath region of $\RR'$. Let $R'$ be such a Macbeath region and let $E$ be the associated canonical cap.
By Lemma~\ref{lem:sandwich}, we have $E^{1/\sigma} \subseteq C \subseteq E$, where $\sigma$ is a suitable constant. We will show that $R'$ and $E$ satisfy the above properties (i)--(iii). 

 To prove (i), recall that $R$ is a Macbeath region with scaling factor $1/5$ centered at the centroid of the base of $C$. It follows from Lemma~\ref{lem:mac2-var} that $C \subseteq R^{5 (3 d)} = R^{15 d}$. Thus $\vol(R) = \Omega(\vol(C)) = \Omega(v)$. Further, since $R$ is centrally symmetric and centered at a point on the base of $C$, we have $\vol(R \cap C) = \vol(R) / 2$. Thus $\vol(R \cap C) = \Omega(v)$.  Since $C \subseteq E$, it follows that $\vol(R \cap E) = \Omega(v)$.

To prove (ii), note that by Lemma~\ref{lem:cap-exp}, $\vol(E) = O(\vol(E^{1/\sigma}))$. Since $E^{1/\sigma} \subseteq C$, we have $\vol(E^{1/\sigma}) = O(\vol(C)) = O(v)$. Thus $\vol(E) = O(v)$. 

To prove (iii), we use the fact that $\vol(C) = O(v)$ and apply Lemma~\ref{lem:mahler-mac}. It follows that the volume of $R'$ is $\Omega(v')$, where $v' = \eps^{d+1} / v$.

In summary, we have shown that each Macbeath region of $\RR$ is associated with a Macbeath region $R' \in \RR'$ and an associated canonical cap $E$ satisfying properties (i)--(iii). Let $\RR'' \subseteq \RR'$ denote the subset of Macbeath regions associated with some Macbeath region of $\RR$ by the above construction, and let $\EE$ denote the set of associated canonical caps. It follows from disjointness of the members of $\RR$ and properties (i) and (ii) that the number of Macbeath regions of $\RR$ that can be contained in any canonical cap $E \in \EE$ is $O(1)$. Thus $|\RR|$ is asymptotically bounded by $|\RR''|$. By property (iii), the Macbeath regions of  $\RR''$ have volume $\Omega(v')$. By Lemma~\ref{lem:mac-cover}(i), the number of such Macbeath regions is $O(\eps / v') = O(v/\eps^d)$. This completes the proof.
\end{proof}

We note that using the previous theorem, the bound on the total number of disjoint Macbeath regions from~\cite{AFM17a} follows easily.

\begin{corollary}
Let $K \subset \RE^d$ be a convex body in $\gamma$-canonical form for constant $\gamma$. Let $\CC$ be a set of caps of $K$ of width $\Theta(\eps)$, such that the Macbeath regions $M'(x)$ centered at the centroids $x$ of the bases of these caps are disjoint. Then $|\CC| = O(1/\eps^{(d - 1) / 2})$.
\end{corollary}

\begin{proof}
Partition $\CC$ into disjoint sets containing caps of $\CC$ whose volumes differ by a factor of at most $2$. By Lemma~\ref{lem:width-vol}, the volume of any cap in $\CC$ is $\Omega(\eps^d)$ and $O(\eps)$, and so there are are $O(\log\inv{\eps})$ such sets. 
Consider a set $S$ in the partition containing caps of volume between $v$ and $2v$, where $v < \eps^{(d + 1) / 2}$. By Theorem~\ref{thm:vol-sensitive-bound}, $|S| = O\kern-1pt\left(\min\left(\frac{\eps}{v}, \frac{v}{\eps^d}\right)\right) = O\kern-1pt\left(\frac{v}{\eps^d}\right)$. Summing up the cardinalities of all those sets $S$ corresponding to volumes $v < \eps^{(d + 1) / 2}$, we obtain a geometric progression that sums to $O(1 / \eps^{(d - 1) / 2})$.
Next consider a set $S'$ in the partition containing caps of volume between $v$ and $2v$, where $v \ge \eps^{(d + 1) / 2}$. By Theorem~\ref{thm:vol-sensitive-bound}, $|S'| = O\kern-1pt\left(\min\left(\frac{\eps}{v}, \frac{v}{\eps^d}\right)\right) = O\kern-1pt\left(\frac{\eps}{v}\right)$. Summing up the cardinalities of all those sets $S'$ corresponding to volumes $v \ge \eps^{(d + 1) / 2}$, we obtain a geometric progression that also sums to $O(1 / \eps^{(d - 1) / 2})$.
Combining these two cases, which together address the entirety of $\CC$, we see that $|\CC| = O(1/\eps^{(d - 1) / 2})$.
\end{proof}

\section{Combinatorial Complexity}
\label{s:opt}

The bound on the total combinatorial complexity of approximating polytopes established in \cite{AFM17a} is sub-optimal by a factor that is polynomial in $\log\inv{\eps}$. In this section, we show how to apply the volume-sensitive bounds developed in the previous sections to eliminate this overhead and obtain an optimal bound. The key idea is to use Macbeath regions centered at points $x$ that are closer than $\eps$ to the boundary of $K$ when the volume of the corresponding Macbeath region is very small or very large. The reason we can afford to do so is that our volume-sensitive bounds show that their numbers are low, so we can afford the increase in their number that may come about through the use of thinner Macbeath regions. Overall, we can adjust parameters to maintain the asymptotic bound of $O(1/\eps^{(d-1)/2})$ on the total number of Macbeath regions. By using thinner Macbeath regions, we can house them in thinner layers during the construction. As the approximation quality of this method is determined by the total thickness of all the layers, our new strategy allows us to improve the approximation quality. To be precise, we can reduce the total thickness of all the layers from $O(\eps \log\inv{\eps})$ in the construction of \cite{AFM17a} to $O(\eps)$ which, in turn, allows us to eliminate the $\polylog(1/\eps)$ overhead and obtain an optimal bound for the combinatorial complexity.

\subsection{Cap Covering} \label{s:cap-cover}

In this subsection, we present a variant of the cap covering lemma proved in \cite{AFM17a}. We will employ this lemma later in establishing an optimal bound on the combinatorial complexity of an approximating polytope. A novel feature of this variant is that the canonical caps in the cover have different widths. 

Before presenting the lemma, we need some definitions. Throughout this section, we assume that the convex body $K$ is in $\gamma$-canonical form for constant $\gamma$. For any integer $j$ (which may be positive, zero, or negative), define:
\[
    v_j ~ = ~ 2^j \eps^{(d+1)/2}, \qquad
    a_j ~ = ~ \max(j^2,1), \qquad\hbox{and}\qquad
    w_j ~ = ~ \frac{\eps}{a_j}. 
\]
We say that a cap $C$ of $K$ is of \emph{type $j$} if $v_j \le \vol(C) < 2v_j$. By Lemma~\ref{lem:width-vol}, the volume of any $\eps$-width cap of $K$ is at least $\Omega(\eps^d)$ and at most $O(\eps)$. It follows that an $\eps$-width cap is of type $j$, where $j$ takes integral values ranging from $-O(\log\inv{\eps})$ to $+O(\log\inv{\eps})$. 

We wish to limit consideration to caps whose volumes and widths satisfy a particular relationship. We say that a set of caps $\CC$ is \emph{$\eps$-balanced} if there exist constants $b_1$ and $b_2$ such that, for all $j$, any type-$j$ cap $C \in \CC$ satisfies $b_1 w_j \le \width(C) \le b_2 w_j$. More precisely, we say that $\CC$ is $(\eps,b_1,b_2)$-balanced. Note, for example, that balanced caps of type 0 have volume $\Theta(\eps^{(d+1)/2})$ and width $\Theta(\eps)$. Up to constant factors, balanced caps of type $j$ (whose volumes differ from type-0 caps by a factor of $2^j$) have widths that are smaller by a factor of $j^2$. 

Our next lemma establishes the key ingredients of our approximation system. It shows that for any convex body $K$ in canonical form there exists a small collection of bodies (lying within $K$ and close to its boundary) and an associated set of $\eps$-balanced caps, such that the bodies form a packing of $K$, the caps cover the boundary of $K$, and each body and its corresponding cap have similar shapes. This generalizes the concept of the \emph{economical cap cover} of \cite{AFM17a} to the context of balanced caps.

\begin{lemma} \label{lem:new-cap-cover}
Let $K$ be a convex body in $\gamma$-canonical form for constant $\gamma$. There exists a set $\RR$ of disjoint, centrally symmetric convex bodies $R_1,\ldots, R_k$, and a collection $\CC$ of associated $\eps$-balanced caps $C_1,\ldots,C_k$ such that the following hold for some constant $\sigma$ (depending only on $d$ and $\gamma$):
\begin{enumerate}
\item[$(1)$] $k = O(1/\eps^{(d-1)/2})$.
\item[$(2)$] For each $i$, $R_i \subseteq C_i \subseteq R_i^{\sigma}$. 
\item[$(3)$] For any direction $u$, there is a cap $C$ whose base is orthogonal to $u$, and which satisfies $R_i \subseteq C$ and $C_i^{1/\sigma} \subseteq C \subseteq C_i$, for some $i$.
\end{enumerate}
\end{lemma} 

We will establish each of the properties separately. Let $\AA = \{A_1, \ldots, A_k\}$ be any maximal set of $(\eps,b_1,b_2)$-balanced caps, such that the Macbeath regions $M'(x_i)$ centered at the centroids $x_i$ of the bases of the caps $A_i^{1/\beta}$ are disjoint. Here $\beta$ is the constant of Lemma~\ref{lem:mac-mac-var}, and $b_1$ and $b_2$ are constants to be specified later (in the proof of Lemma~\ref{lem:balanced-cap}). Let $R_i = M'(x_i)$ and $C_i = A_i^{\beta}$, and let $\RR, \CC$, and $\AA'$ be the sets consisting of $R_i, C_i$, and $A_i^{1/\beta}$, respectively, for $1 \le i \le k$. It follows from Lemma~\ref{lem:cap-exp} that constant factor expansions of $\eps$-balanced caps are $\eps$-balanced, so the sets $\CC$ and $\AA'$ are also $\eps$-balanced (for different constants $b_1$ and $b_2$).

In Lemma~\ref{lem:cover-size}, we establish Property~1 by showing a bound on $|\AA'|$. In Lemma~\ref{lem:bal-cap-prop2}, we establish Property~2, and in Lemmas~\ref{lem:balanced-cap} and \ref{lem:bal-cap-prop3}, we establish  Property~3.

\begin{lemma} \label{lem:cover-size}
The number of $\eps$-balanced caps such that the Macbeath regions $M'(x)$ centered at the centroids $x$ of the bases of these caps are disjoint is $O(1/\eps^{(d-1)/2})$.
\end{lemma}

\begin{proof}
First we bound the number of such caps of type $j$ where $j \ge 0$. As these caps are $\eps$-balanced, they have width $\Theta(w_j)$. By Theorem~\ref{thm:vol-sensitive-bound}, their number $n_j$ is $O(w_j / v_j)$. Recalling that $v_j = \eps^{(d+1)/2} \cdot 2^j$ and $w_j = \eps/\max(j^2,1)$, we have
\[
    n_j 
        ~ = ~ O\left( \frac{\eps/\max(j^2,1)}{\eps^{\frac{d+1}{2}} \cdot 2^j} \right)
        ~ = ~ O\left( \frac{1}{\eps^{\frac{d-1}{2}}} \kern-1pt \cdot \kern-1pt \frac{1}{2^j \max(j^2,1)} \right).
\]
Summing $n_j$ over all $j \ge 0$, we obtain a total of $N_+ = O(1/\eps^{(d-1)/2})$ caps.

Next we bound the number of such caps of type $j$ where $j < 0$. By Theorem~\ref{thm:vol-sensitive-bound}, their number $n_j$ is $O(v_j / w_j^d)$. Thus
\[
    n_j 
        ~ = ~ O\left( \frac{\eps^{(d+1)/2} \cdot 2^j}{(\eps/j^2)^d} \right) 
        ~ = ~ O\left( \frac{1}{\eps^{(d - 1)/2}} \cdot 2^j j^{2d} \right).
 \] 
Summing $n_j$ over all $j < 0$, we obtain a total of $N_- = O(1/\eps^{(d-1)/2})$ caps. Therefore, the total number of caps, $N_+ + N_- = O(1/\eps^{(d-1)/2})$, as desired.
\end{proof}

\begin{lemma} \label{lem:bal-cap-prop2}
There exists a constant $\sigma$ such that for each $i$, $R_i \subseteq C_i \subseteq R_i^{\sigma}$.
\end{lemma}

\begin{proof}
Recall that $R_i = M'(x_i)$. By Lemma~\ref{lem:mac-cap2}, the expansion $A_i^{2/\beta}$ will fully contain $R_i$. Since $A_i^{2/\beta} \subseteq A_i^{\beta} = C_i$, we obtain $R_i \subseteq C_i$. By Lemma~\ref{lem:mac2-var-exp}, 
\[
    C_i 
        ~ =         ~ (A_i^{1/\beta})^{\beta^2} 
        ~ \subseteq ~ M^{3d(2\beta^2 - 1)}(x_i) 
        ~ =         ~ R_i^{15d(2\beta^2 - 1)} 
        ~ \subseteq ~ R_i^{\sigma},
\]
for any constant $\sigma \ge 15d(2\beta^2 - 1)$. Thus, $R_i \subseteq C_i \subseteq R_i^{\sigma}$, as desired.
\end{proof}

\begin{lemma} \label{lem:balanced-cap}
Given any $\eps$-width cap $C$ of type $j$, the cap $C^{1/a_j}$ is a $(\eps,b_1,b_2)$-balanced cap for suitable constants $b_1$ and $b_2$ (which do not depend on $j$). 
\end{lemma}

\begin{proof}
Since $C$ is of type $j$, we have $\eps^{(d+1)/2} 2^j \le \vol(C) < \eps^{(d+1)/2} 2^{j+1}$. Let $C' = C^{1/a_j}$. Clearly, $\vol(C') \le \vol(C)$ and, by Lemma~\ref{lem:cap-exp}, $\vol(C) \le a_j^d \cdot \vol(C')$. Thus
\[
    \eps^{\frac{d+1}{2}} 2^{j+1} 
        ~ >    ~ \vol(C') 
        ~ \geq ~ \frac{\eps^{\frac{d+1}{2}} 2^j}{a_j^d}
        ~ =    ~ \eps^{\frac{d+1}{2}} 2^{j - d \log(a_j)}.
\]
Letting $k$ denote the type of $C'$, we have $\eps^{(d+1)/2} 2^k \le \vol(C') < \eps^{(d+1)/2} 2^{k+1}$. These inequalities readily imply that
\[
    j + 1 
        ~ > ~ k 
        ~ > ~ j - d \log(a_j) - 1.
\]
It is easy to see that $a_j = \Theta(a_k)$. As the width of $C'$ is $\eps/a_j$, which is $\Theta(\eps/a_k)$, it follows that we can choose $b_1$ and $b_2$ such that $C'$ is $(\eps,b_1,b_2)$-balanced.
\end{proof}

\begin{lemma} \label{lem:bal-cap-prop3}
There exists a constant $\sigma$ such that the following holds. For any direction $u$, there is a cap $C$ whose base is orthogonal to $u$, and which satisfies $R_i \subseteq C$ and $C_i^{1/\sigma} \subseteq C \subseteq C_i$, for some $i$.
\end{lemma}

\begin{proof}
Let $F$ be an $\eps$-width cap whose base is orthogonal to direction $u$. Suppose that $F$ is of type $j$. We will show that the cap $C = F^{1/a_j}$ satisfies the properties given in the statement of the lemma. 

By Lemma~\ref{lem:balanced-cap}, $C$ is $(\eps,b_1,b_2)$-balanced. Let $R = M'(x)$ be the Macbeath region centered at the centroid $x$ of the base of the cap $C^{1/\beta}$. By our construction, there must exist $i$ such that the Macbeath region $R_i$ intersects $R$. Recall that $C_i = A_i^{\beta}$ and $R_i = M'(x_i)$, where $x_i$ is the centroid of the base of the cap $A_i^{1/\beta}$. Since $M'(x_i) \cap M'(x) \neq \emptyset$, by Lemma~\ref{lem:mac-mac}, $M'(x_i) \subseteq M(x)$. Also, by Lemma~\ref{lem:mac-cap2-var}, $M(x) \subseteq C^{2/\beta}$. Clearly $C^{2/\beta} \subseteq C$. Putting it together, we have
$R_i = M'(x_i) \subseteq M(x) \subseteq C^{2/\beta} \subseteq C$, which proves the first part of the lemma.

It remains to show that $C_i^{1/\sigma} \subseteq C \subseteq C_i$. Since $M'(x_i) \cap M'(x) \neq \emptyset$, we can apply Lemma~\ref{lem:mac-mac-var} to caps $A_i^{1/\beta}$ and $C^{1/\beta}$ (for $\rho = 1$) to obtain $A_i^{1/\beta} \subseteq C$. Applying Lemma~\ref{lem:mac-mac-var} again to caps $C^{1/\beta}$ and $A_i^{1/\beta}$ (for $\rho = \beta$), we obtain $C \subseteq A_i^{\beta}$. Recalling that $C_i = A_i^{\beta}$, we have $C_i^{1/\beta^2} \subseteq C \subseteq C_i$. The claim now follows for any positive constant $\sigma \ge \beta^2$.
\end{proof}

Lemma~\ref{lem:new-cap-cover} follows by setting $\sigma$ to the maximum of the constants used in Lemmas~\ref{lem:bal-cap-prop2} and~\ref{lem:bal-cap-prop3}.

\subsection{Witness-Collector Technique} \label{s:witness}

In this section, we will provide a quick overview of the witness-collector approach~\cite{AFM17a}, which is central to our construction. Recall that $K$ is a convex body in $\gamma$-canonical form for some constant $\gamma$. The general strategy is as follows. First, we build a set $\RR$ of disjoint, centrally symmetric convex bodies lying within $K$ and close to its boundary. These bodies will possess certain key properties to be specified later. For each $R \in \RR$, we select a point arbitrarily from this body, and let $S$ denote this set of points. The approximation $P$ is defined as the convex hull of $S$. In Lemma~\ref{lem:apx}, we will prove that $P$ is an $\eps$-approximation of $K$ and, in Lemma~\ref{lem:fewfaces}, we will apply a deterministic variant of the witness-collector approach~\cite{DGGT16} to show that $P$ has low combinatorial complexity.

Let us recall the basic definitions from the introduction. We define a set $\WW$ of regions called \emph{witnesses} and a set $\CC$ of regions called \emph{collectors}, which satisfy the following properties:
\begin{enumerate}
\item[(1)] Each witness of $\WW$ contains a point of $S$ in its interior.

\item[(2)] Any halfspace $H$ either contains a witness $W \in \WW$ or $H \cap S$ is contained in a collector $C \in \CC$.

\item[(3)] Each collector $C \in \CC$ contains at most a constant number of points of $S$.
\end{enumerate}

The key idea of the witness-collector method is encapsulated in the following lemma, which was proved in~\cite{AFM17a}. For the sake of completeness, we repeat the proof here.

\begin{lemma} \label{lem:witness-collector}
Given a set of witnesses and collectors satisfying the above properties, the combinatorial complexity of the convex hull $P$ of $S$ is $O(|\CC|)$.
\end{lemma}

\begin{proof}
We map each face $f$ of the convex hull of $P$ to any maximal subset $S_f \subseteq S$ of affinely independent points on $f$. Note that this is a one-to-one mapping and $|S_f| \le d$. In order to bound the combinatorial complexity of $P$ it suffices to bound the number of such subsets $S_f$. 

For a given face $f$, let $H$ be any halfspace such that $H \cap P = f$. Clearly $H$ does not contain any witness since otherwise, by Property~1, it would contain a point of $S$ in its interior. By Property~2, $H \cap S$ is contained in some collector $C \in \CC$. Thus $S_f \subseteq C$. Since $|S_f| \le d$, it follows that the number of such subsets $S_f$ that are contained in any collector $C$ is at most 
\[
    \sum_{1 \leq j \leq d} {\binom{|C \cap S|}{j}} 
	    ~ = ~ O(|C \cap S|^d) 
	    ~ = ~ O(1),
\]
where in the last step we have used the fact that $|C \cap S|  = O(1)$ (Property~3). Summing over all the collectors, it follows that the total number of sets $S_f$, and hence the combinatorial complexity of $P$, is $O(|\CC|)$.
\end{proof}

\subsection{Layered Construction} \label{s:layering}

A natural choice for the witnesses and collectors would be the convex bodies $R_i$ and the caps $C_i$, respectively, from Lemma~\ref{lem:new-cap-cover}. As shown in~\cite{AFM17a}, these bodies do not work for our purposes. The main difficulty is that Property~(3) could fail, since a cap $C_i$ could intersect a non-constant number of bodies of $\RR$, and hence contain a non-constant number of points of $S$. Following the general approach of that earlier paper, we show that it is possible to construct a set of witnesses and collectors that satisfy all the requirements by scaling and translating the convex bodies from Lemma~\ref{lem:new-cap-cover} into layers according to their volumes. The properties we obtain are specified below in Lemma~\ref{lem:layers}.

Our choice of witnesses and collectors will be based on the following lemma. Specifically, the convex bodies $R_1, \ldots, R_k$ given in the statement of the lemma, will play the role of the witnesses and the regions $C_1, \ldots, C_k$ (not to be confused with caps) will play the role of the collectors. The lemma strengthens Lemma~\ref{lem:new-cap-cover}, achieving the critical property that any collector $C_i$ intersects only a constant number of convex bodies of $\RR$. As each witness set $R_i$ will contain one point, this ensures that a collector contains only a constant number of input points (Property~(3) of the witness-collector system). This strengthening is achieved while maintaining the same number of collectors asymptotically, as in Lemma~\ref{lem:new-cap-cover}. Also, the collectors are no longer simple caps, but have a more complex shape as described in the proof (this, however, has no adverse effect in our application). 

\begin{lemma} \label{lem:layers}
Let $\eps > 0$ be a sufficiently small parameter. Let $K \subset \RE^d$ be a convex body in canonical form.  There exists a set $\RR$ of $k = O(1/\eps^{(d-1)/2})$ disjoint, centrally symmetric convex bodies $R_1, \ldots, R_k$ and a collection $\CC$ of associated regions $C_1,\ldots, C_k$ such that the following hold:
\begin{enumerate}
\item Let $C$ be any cap of width $\eps$. Then there is an $i$ such that $R_i \subseteq C$.

\item Let $C$ be any cap. Then there is an $i$ such that either (i) $R_i \subseteq C$ or (ii) $C \subseteq C_i$.

\item For each $i$, the region $C_i$ intersects at most a constant number of bodies of $\RR$.
\end{enumerate}
\end{lemma}

As mentioned earlier, our proof of this lemma is based on a layered placement of the convex bodies from Lemma~\ref{lem:new-cap-cover}, which are distributed among $O(\log\inv{\eps})$ layers that lie close to the boundary of $K$. Let $\alpha = c_0 \, \eps$, where $c_0$ is a suitable constant to be specified later. We begin by applying Lemma~\ref{lem:new-cap-cover} to $K$ using $\eps = \alpha$. This yields a collection $\RR'$ of $k = O(1/\alpha^{(d-1)/2})$ disjoint, centrally symmetric convex bodies $\{R'_1, \ldots, R'_k\}$ and associated caps $\CC' = \{C'_1, \ldots, C'_k\}$. Our definition of the convex bodies $R_i$ and regions $C_i$ required in Lemma~\ref{lem:layers} will be based on $R'_i$ and $C'_i$, respectively. In particular, the convex body $R_i$ will be obtained by translating a scaled copy of $R'_i$ into an appropriate layer, based on the type of the cap $C'_i$. 

Recall that the caps of $\CC'$ are $(\alpha,b_1,b_2)$-balanced for some constants $b_1$ and $b_2$, and have integral types ranging from $-t$ to $+t$, where $t = O(\log\inv{\alpha})$. By definition, the width of any cap of type $j$ in $\CC'$ is at most $b_2 w_j$, where $w_j = \alpha/\max(j^2,1) = c_0 \eps / \max(j^2,1)$. 

We partition the set $\RR'$ of convex bodies into $2t+1$ groups, based on the type of the associated cap $C'$. More precisely, for $-t \le j \le t$, \emph{group $j$} consists of bodies $R' \in \RR'$, such that the associated cap $C'$ is of type $j$. 

Next we describe how the layers are constructed. We will construct $2t+1$ layers corresponding to the $2t+1$ groups of $\RR'$. Our construction uses a constant parameter $c_1$. For $-t-1 \le j \le t-1$, let $T_j$ denote the linear transformation that represents a uniform scaling by a factor of $\prod_{i=j+1}^t (1 - c_1 w_i)$ about the origin, and let $T_t$ denote the identity transformation. Further, define $K_j = T_j(K)$ (see Figure~\ref{f:layers}(a)) and let $K_t = K$. For $-t \le j \le t$, define layer $j$, denoted $L_j$, to be the difference $K_j \setminus K_{j-1}$. Whenever we refer to parallel supporting hyperplanes for two bodies $K_i$ and $K_j$, we assume that both hyperplanes lie on the same side of the origin.

\begin{figure*}[tbp]
  \centerline{\includegraphics[scale=.65]{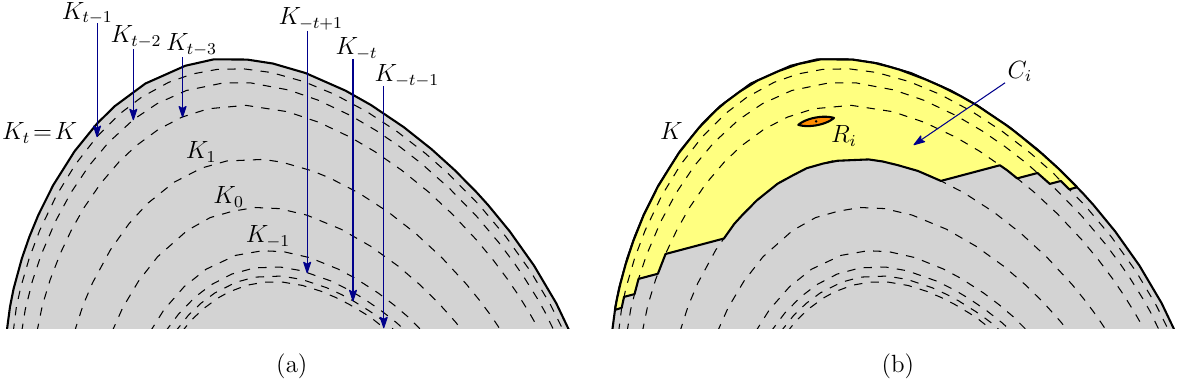}}
  \caption{\label{f:layers}(a) The layers of $K$ in the layered construction and (b) the region $C_i$ corresponding to a body $R_i$.}
\end{figure*}

The following lemma describes some straightforward properties of these layers and the scaling transformations. In particular, the lemma shows that the $t$ layers lie close to the boundary of $K$ (within distance $\eps$) and layer $j$ has a ``thickness'' of $\Theta(w_j)$.

\begin{lemma} \label{lem:scale}
Let $\eps > 0$ be a sufficiently small parameter and $c_1$ be any constant. For sufficiently small constant $c_0$ in the definition of $\alpha$ (depending on $c_1)$, the layered construction and the scaling transformations described above satisfy the following properties:
 \begin{enumerate}[label=(\alph*)]
 \item\label{a} For $-t \le j \le t$, the distance between parallel supporting hyperplanes of $K_{j-1}$ and $K_j$ is at most $c_1 w_j / \sqrt{\gamma}$.
 
 \item\label{b}  For $-t \le j \le t$, the distance between parallel supporting hyperplanes of $K_{j-1}$ and $K_j$ is at least $\sqrt{\gamma} c_1 w_j / 2$.

 \item\label{c} The distance between parallel supporting hyperplanes of $K$ and $K_{-t-1}$ is at most $\eps$.
  
 \item\label{d} For $-t-1 \le j \le t$, the scaling factor for $T_j$ is at least 1/2 and at most 1.
  
 \item\label{e} For $-t-1 \le j \le t$, $T_j$ preserves volumes up to a constant factor.
\end{enumerate}

\end{lemma}

\begin{proof}
To prove \ref{a}, let $h_1,h_2$ denote parallel supporting hyperplanes of $K_j,K_{j-1}$, respectively. Since $K$ is in $\gamma$-canonical form, and the scaling factor of the transformation $T_j$ is at most 1, it follows that $h_1$ is at distance at most $1/\sqrt{\gamma}$ from the origin.  Since $h_2$ is obtained by scaling $h_1$ by a factor of $1 - c_1 w_j$ about the origin, it follows that the distance between $h_1$ and $h_2$ is at most $c_1 w_j / \sqrt{\gamma}$.

To prove \ref{c}, let $h_1,h_2$ denote parallel supporting hyperplanes of $K, K_{-t-1}$, respectively. The upper bound of \ref{a} implies that the distance between $h_1$ and $h_2$ is at most $\sum_{j=-t}^t c_1 w_j / \sqrt{\gamma}$. Recall that $w_j = \alpha / \max(j^2,1)$, where $\alpha = c_0 \eps$. By choosing a sufficiently small constant $c_0$ in the definition of $\alpha$ (depending on $c_1$ and $\gamma$), we can ensure that the distance between $h_1$ and $h_2$ is at most $\eps$. 

In the rest of this proof, we will assume that $c_0$ in the definition of $\alpha$ is sufficiently small, so \ref{c} holds. To prove \ref{b}, let $h_1,h_2$ denote parallel supporting hyperplanes of $K_j,K_{j-1}$, respectively. Since $K$ is in $\gamma$-canonical form, it follows from (c) that $h_1$ is at distance at least $\sqrt{\gamma} - \eps$ from the origin. Since $h_2$ is obtained by scaling $h_1$ by a factor of $1 - c_1 w_j$ about the origin, it follows that the distance between $h_1$ and $h_2$ is at least $c_1 (\sqrt{\gamma} - \eps) w_j$, which is at least $\sqrt{\gamma} c_1 w_j / 2 $ for $\eps \le \sqrt{\gamma}/2$. 

To prove \ref{d}, note that we only need to show the lower bound on the scaling factor of $T_j$, since the upper bound is obvious. Again, let $h_1,h_2$ denote parallel supporting hyperplanes of $K, K_{-t-1}$, respectively. Since $K$ is in $\gamma$-canonical form, $h_1$ is at distance at least $\sqrt{\gamma}$ from the origin. Recall that $T_{-t-1}$ maps $h_1$ to $h_2$ and, as shown above, the distance between $h_1$ and $h_2$ is at most $\eps$. It follows that the scaling factor of $T_{-t}$ is at least $1 - \eps / \sqrt{\gamma}$. By choosing $\eps \le \sqrt{\gamma}/2$, we can ensure that the scaling factor of $T_{-t-1}$ is at least $1/2$. Clearly, this lower bound on the scaling factor also applies to any transformation $T_j$, $-t-1 \le j \le t$. This proves \ref{d}. Note that \ref{e} is an immediate consequence.
\end{proof}

Recall that the width of a cap of type $j$ in $\CC'$ is at most $b_2 w_j$ for some constant $b_2$. In order to ensure that layer $j$ can accommodate caps of type $j$, we construct the layered construction of Lemma~\ref{lem:scale} for a constant $c_1 = 2 b_2/\sqrt{\gamma}$. This choice ensures that the distance between parallel supporting hyperplanes of $K_{j-1}$ and $K_j$, respectively, is at most $(2b_2/\gamma) w_j$ and at least $b_2 w_j$ (properties \ref{a} and \ref{b} in Lemma~\ref{lem:scale}).

Let $H'$ be a halfspace and let $C' = K \cap H'$ be a type-$j$ cap in $\CC'$. It will be convenient to associate a set of caps with $C'$ that occur frequently in our construction and analysis. For $j \le r \le t$, define $E_r = K_r \cap T_j(H')$ and define $F_r = T_r(C')$. Both $E_r$ and $F_r$ are caps of $K_r$. This is obviously true for $E_r$. To see that $F_r$ is a cap of $K_r$, note that $F_r = T_r(C') = T_r(K \cap H') = T_r(K) \cap T_r(H') = K_r \cap T_r(H')$. Also, note that $E_j = F_j$.

We are now ready to define the sets $\RR$ and $\CC$ required in Lemma~\ref{lem:layers}. Let $R' \in \RR'$ be a body in group $j$ and let the cap in $\CC'$ associated with it be $C' = K \cap H'$. We define a body $R \in \RR$ and an associated region $C \in \CC$ corresponding to $R'$ and $C'$ as
\begin{equation} \label{eq:RCdef}
    R ~ = ~ T_j(R') \quad\text{and}\quad
    C ~ = ~ \bigcup_{r = j}^{t} E_r^{\sigma} \cap L_r,
\end{equation}
where $\sigma$ is the constant of Lemma~\ref{lem:new-cap-cover} (see Figure~\ref{f:layers}(b)). 

In Lemma~\ref{lem:strat-prop0}, we show that the regions $R$ are contained in layer $j$ if $R'$ is in group $j$. In Lemmas~\ref{lem:strat-prop12-aux} and \ref{lem:strat-prop12}, we establish Properties 1 and 2 of Lemma~\ref{lem:layers}. Finally, in Lemmas~\ref{lem:cap-vol-poly} and \ref{lem:strat-prop3}, we establish Property~3 of Lemma~\ref{lem:layers}.

\begin{lemma} \label{lem:strat-prop0}
If $R'$ is in group $j$, then $R \subseteq F_j \subseteq L_j$.
\end{lemma}

\begin{proof}
By Property~2 of Lemma~\ref{lem:new-cap-cover}, $R' \subseteq C'$. Applying the transformation $T_j$ to these two sets yields $R \subseteq F_j$. Next we show that $F_j \subseteq L_j$. Since $C'$ is a cap of type $j$, its width is at most $b_2 w_j$. By Lemma~\ref{lem:scale}\ref{d}, the scaling factor for $T_j$ is at most 1. Thus, the width of $F_j$ is at most $b_2 w_j$. By Lemma~\ref{lem:scale}\ref{b} and our remarks following Lemma~\ref{lem:scale}, the distance between any parallel supporting hyperplanes of $K_{j-1}$ and $K_j$, respectively, is at least $b_2 w_j$. It follows that $F_j \subseteq K_j \setminus K_{j-1} = L_j$. This completes the proof.
\end{proof}

\begin{lemma} \label{lem:strat-prop12-aux}
For any direction $u$, there is a cap $A$ whose base is orthogonal to $u$, and which satisfies $R \subseteq A \subseteq C$, for some $R \in \RR$ and $C \in \CC$. Further, the width of the cap $A$ is at most $\eps$.
 \end{lemma}

\begin{proof}
By Property~3 of Lemma~\ref{lem:new-cap-cover}, there exists a cap $\widehat{A} = K \cap \widehat{H}$ whose base is orthogonal to $u$, and which satisfies $R' \subseteq \widehat{A}$ and $(C')^{1/\sigma}  \subseteq \widehat{A} \subseteq C'$ for some $R' \in \RR'$ and $C' \in \CC'$. Let $C'$ be a type-$j$ cap. Define $H = T_j(\widehat{H})$ and $A = K \cap H$. We will show that the cap $A$ possesses the properties given in the statement of the lemma. See Figure~\ref{f:prop12} for a representation of the definitions.

\begin{figure}[htb]
  \centerline{\includegraphics[trim={0 0 1.8cm 0},clip,scale=.65]{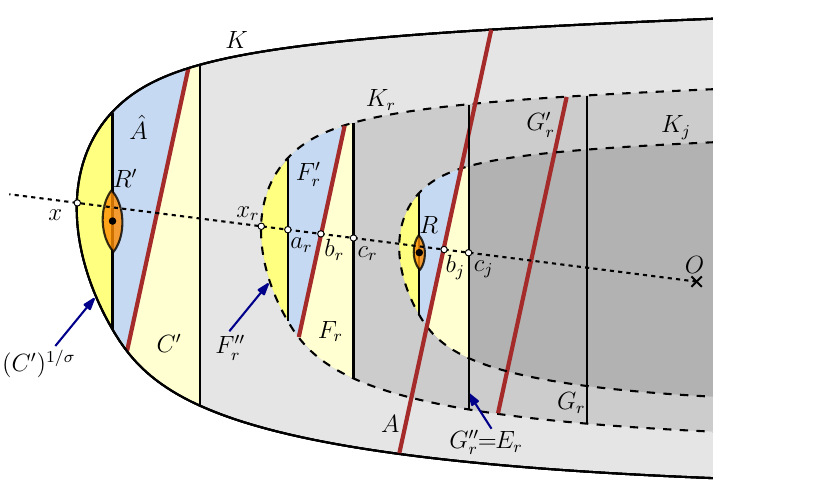}}
  \caption{\label{f:prop12}Proof of Lemma~\ref{lem:strat-prop12-aux}.}
\end{figure}

Recall the body $R \in \RR$ and region $C \in \CC$ (defined in Eq.~\eqref{eq:RCdef}), corresponding to $R'$ and $C'$. Since $R' \subseteq \widehat{A} = K \cap \widehat{H}$, we can apply the transformation $T_j$ to these sets to obtain $R \subseteq K_j \cap H \subseteq K \cap H = A$.

Next we show that the width of the cap $A$ is at most $\eps$. Recall that $\widehat{A} \subseteq C'$. Applying the transformation $T_j$ to these sets, we obtain $K_j \cap H \subseteq F_j$. By Lemma~\ref{lem:strat-prop0}, $F_j \subseteq L_j$. Thus $K_j \cap H \subseteq L_j$. Also, by Lemma~\ref{lem:scale}\ref{c}, the distance between any parallel supporting hyperplanes of $K$ and $K_{-t-1}$ is at most $\eps$. Since $K_j \cap H \subseteq L_j$, it follows that the width of the cap $A = K \cap H$ is at most $\eps$. 

It remains to show that $A \subseteq C$. By the definition of $A$ and $C$, it suffices to show that $K_j \cap H \subseteq E_j^{\sigma} \cap L_j$ and, for $j+1 \le r \le t$, $K_r \cap H ~\subseteq~ E_r^{\sigma}$. Applying $T_j$ to both sides of $\widehat{A} \subseteq C'$, we obtain $K_j \cap H ~\subseteq~ F_j = E_j \subseteq E_j^{\sigma}$ (for any $\sigma \ge 1$). By Lemma~\ref{lem:strat-prop0}, $F_j \subseteq L_j$. Thus $K_j \cap H \subseteq E_j^{\sigma} \cap L_j$.

Next we consider the case when $r \geq j+1$.
In this case, we will need to exploit the fact that $\widehat{A}$ is sandwiched between two caps with parallel bases, that is, $(C')^{1/\sigma}  \subseteq \widehat{A} \subseteq C'$. Recall that $F_r$ and $F_j$ are the caps of $K_r$ and $K_j$, respectively, defined as $F_r = T_r(C')$ and $F_j = T_j(C')$. Define $F'_r = T_r(\widehat{A}), F'_j = T_j(\widehat{A})$, and $F''_r = T_r((C')^{1/\sigma})$. We have $F''_r \subseteq F'_r \subseteq F_r$ and $F'_j \subseteq F_j$.

Let $x$ denote the apex of $C'$ and $x_r$ denote the point $T_r(x)$. Let $a_r, b_r$, and $c_r$ denote the points of intersection of the bases of the caps $F''_r, F'_r$, and  $F_r$, respectively, with the line segment $O x$. Similarly, let $b_j$ and $c_j$ denote the points of intersection of the bases of the caps $F'_j$ and $F_j$, respectively, with the line segment $O x$. Consider scaling caps $F''_r, F'_r$ and $F_r$ as described in Lemma~\ref{lem:cap-containment}, about the point $x_r$ with scaling factor $\|c_j x_r\| / \|a_r x_r\|$. Let $G''_r, G'_r,$ and $G_r$ denote the caps of $K_r$ obtained from $F''_r, F'_r,$ and $F_r$, respectively, through this transformation. By Lemma~\ref{lem:cap-containment}, $G''_r \subseteq G'_r \subseteq G_r$. Note that $F''_r$ and $E_r$ are caps of $K_r$ with parallel bases and the base of $E_r$ passes through the point $c_j$ (since the hyperplanes passing through the bases of the caps $E_r$ and $F_j$ are the same). Also, $K_r \cap H$ and $G'_r$ are caps of $K_r$ with parallel bases. Our choice of the scaling factor thus implies that $G''_r = E_r$ and $K_r \cap H \subseteq G'_r$.

Putting these facts together, we have $K_r \cap H  \subseteq G_r$. Note that $G_r$ and $E_r$ are caps of $K_r$ with parallel bases. Thus, to prove that $K_r \cap H \subseteq E_r^{\sigma}$, it suffices to show that $\width(G_r) / \width(E_r) \le \sigma$. Note that 
\[
    \frac{\width(G_r)}{\width(E_r)} 
        ~ = ~ \frac{\width(G_r)}{\width(G''_r)} 
        ~ = ~ \frac{\width(F_r)}{\width(F''_r)} 
        ~ = ~  \frac{\width(C')}{\width((C')^{1/\sigma})} 
        ~ = ~ \sigma,
\]
which completes the proof.
\end{proof}

\begin{lemma} \label{lem:strat-prop12}
Let $A$ be any cap of $K$. Then either (i) there is a body $R \in \RR$ such that $R \subseteq A$ or (ii) there is a region $C \in \CC$ such that $A \subseteq C$. Furthermore, if the width of $A$ is $\eps$, then (i) holds.
\end{lemma}

\begin{proof}
Taking $u$ to be the unit vector orthogonal to the base of the cap $A$ and applying Lemma~\ref{lem:strat-prop12-aux}, it follows that there exists a cap $A'$ whose base is parallel to the base of $A$ and which satisfies $R \subseteq A' \subseteq C$, for some $R \in \RR$ and $C \in \CC$. Further, the width of the cap $A'$ is at most $\eps$.

We consider two cases, depending on whether $A' \subseteq A$ or $A \subseteq A'$. In the first case, we have $R \subseteq A' \subseteq A$ and, in the second case, we have $A \subseteq A' \subseteq C$. Thus, either $R \subseteq A$ or $A \subseteq C$.

Further, if the width of $A$ is $\eps$, then $A' \subseteq A$ because the width of $A'$ is at most $\eps$. Thus the first case holds implying that $R \subseteq A$.
\end{proof}

In Lemma~\ref{lem:strat-prop3}, we bound the number of bodies of $\RR$ that overlap any region $C \in \CC$ (Property~3 of Lemma~\ref{lem:layers}). Recall that $C$ corresponds to a cap $C' \in \CC'$ (by Eq.~\eqref{eq:RCdef}). Let $C'$ be of type $j$. We first establish a constant bound on the number of bodies of $\RR$ that overlap $E_j^{\sigma} \cap L_j$. Then we bound the number of bodies of $\RR$ that overlap $E_r^{\sigma} \cap L_r$ for $r > j$. Our analysis exploits the fact that the volume of $E_r$ exceeds the volume of $E_j$ by a factor that is at most polynomial in $r-j$ (i.e., the number of layers between $K_r$ and $K_j$), while the volume of the bodies of $\RR$ in layer $r$ exceeds the volume of the bodies of $\RR$ in layer $j$ by a factor that is exponential in $r-j$. This allows us to show that the number of bodies of $\RR$ that overlap $C$ is bounded by a constant. 

Before presenting Lemma~\ref{lem:strat-prop3}, we establish a polynomial bound on the growth rate of the volume of the caps $E_r$ in the following lemma.

\begin{lemma} \label{lem:cap-vol-poly}
Let $C'$ be a type-$j$ cap. For $j+1 \le r \le t$, $\vol(E_r) = O((r-j)^{3d}) \cdot \vol(E_j)$.
\end{lemma}

\begin{proof}
Recall that $F_r = T_r(C')$ and $E_j = F_j = T_j(C')$. By Lemma~\ref{lem:scale}\ref{e}, $T_j$ and $T_r$ preserve volumes up to constant factors, and so $\vol(F_r) = \Theta(\vol(F_j))$. Thus, to prove the lemma, it suffices to show that $\vol(E_r) / \vol(F_r) = O((r-j)^{3d})$. In turn, in light of Lemma~\ref{lem:cap-exp}, it suffices to prove that $\width(E_r) / \width(F_r) = O((r-j)^3)$.

Towards this end, recall that the width of $E_r$ is upper bounded by the distance between parallel supporting hyperplanes of $K_r$ and $K_{j-1}$ which, by Lemma~\ref{lem:scale}\ref{a}, is at most $O(\sum_{i=j}^r w_i)$. Further, by Lemma~\ref{lem:scale}\ref{d}, the width of $F_r$ is at least half the width of $C'$. As $C'$ is of type $j$, by definition its width is $\Theta(w_j)$. It follows that the width of $F_r$ is $\Omega(w_j)$. Thus, we have shown that 
\[
    \frac{\width(E_r)}{\width(F_r)} 
        ~ = ~ O \left( \frac{\sum_{i=j}^r w_i}{w_j} \right).
\]
Clearly,
\[
    \frac{\sum_{i=j}^r w_i}{w_j} 
        ~ \leq ~ (r-j+1) \cdot \left(\frac{\max_{i=j}^r w_i}{w_j}\right).
\]
To complete the proof, we will show that $(\max_{i=j}^r w_i) / w_j = O((r-j)^2)$. Recall that for any $i$, $w_i = \alpha / \max(i^2,1)$. We consider three cases: (1) $r \ge j \ge 0$, (2) $0 > r \ge j$, and (3) $r \ge 0 > j$. In Case 1, we have $\max_{i=j}^r w_i = w_j$, and so the quantity of interest is 1. In Case 2, $\max_{i=j}^r w_i = w_r$. Thus
\[
    \frac{\max_{i=j}^r w_i}{w_j}
        ~ =    ~ \frac{w_r}{w_j} = \frac{1/r^2}{1/j^2} = \frac{j^2}{r^2} 
        ~ =    ~ \left( \frac{(r-j) + |r|}{|r|} \right)^2 
        ~ =    ~ \left( \frac{r-j}{|r|} + 1 \right)^2
        ~ \leq ~ (r-j+1)^2.
\]
In Case 3, $\max_{i=j}^r w_i = w_0$. Thus
\[
    \frac{\max_{i=j}^r w_i}{w_j} 
        ~ =    ~ \frac{1}{1/j^2} 
        ~ =    ~ j^2 
        ~ \leq ~ (r-j)^2.
\]
In all three cases, we have shown that $(\max_{i=j}^r w_i) / w_j = O((r-j)^2)$, as desired.
\end{proof}

\begin{lemma} \label{lem:strat-prop3}
Any region $C \in \CC$ intersects $O(1)$ bodies of $\RR$.
\end{lemma}

\begin{figure}[tbp]
  \centerline{\includegraphics[scale=.65]{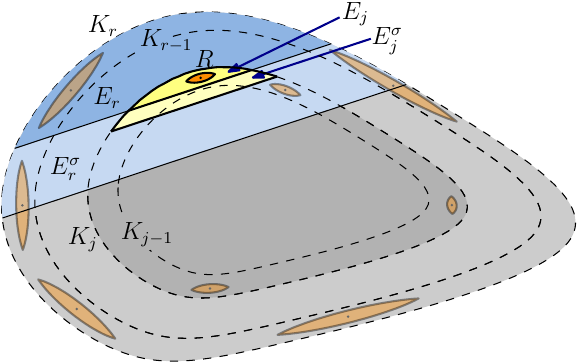}}
  \caption{\label{f:prop3}Proof of Lemma~\ref{lem:strat-prop3}.}
\end{figure}

\begin{proof}
Recall by Eq.~\eqref{eq:RCdef} that $C$ corresponds to a cap $C' \in \CC'$ of some type $j$, where $C = \bigcup_{r = j}^{t} (E_r^{\sigma} \cap L_r)$. We begin by bounding the number of bodies of $\RR$ that overlap $E_j^{\sigma} \cap L_j$ (see Figure~\ref{f:prop3}). We assert that all the bodies of $\RR$ in layer $j$ have volumes $\Omega(\vol(E_j))$. To prove this, recall that the type-$j$ caps of $\CC'$ have the same volume as $C'$ to within a factor of 2. Also, recall that the body of $\RR'$ associated with a cap of $\CC'$ has the same volume as the cap to within a constant factor (immediate consequence of Property~2 of Lemma~\ref{lem:new-cap-cover}). It follows that all the bodies of $\RR'$ in group $j$ have volumes $\Omega(\vol(C'))$. By Lemma~\ref{lem:scale}\ref{e}, the scaling transformations used in our construction preserve volumes to within a constant factor. Also, recall that $E_j$ is a scaled copy of $C'$, and by Lemma~\ref{lem:strat-prop0}, the bodies of $\RR$ in layer $j$ are scaled copies of the bodies of $\RR'$ in group $j$. It follows that the bodies of $\RR$ in layer $j$ all have volumes $\Omega(\vol(E_j))$.

Next, we assert that any body of $\RR$ that overlaps $E_j^{\sigma} \cap L_j$ is contained within the cap $E_j^{2 \sigma}$. To prove this, recall from the proof of Lemma~\ref{lem:new-cap-cover} that the bodies of $\RR'$ are  $(1/5)$-scaled disjoint Macbeath regions with respect to $K$. It follows that the bodies of $\RR$ in layer $j$ are $(1/5)$-scaled disjoint Macbeath regions with respect to $K_j$. Furthermore, by Lemma~\ref{lem:strat-prop0} any body of $\RR$ that overlaps $L_j$ lies entirely within $L_j$. By Lemma~\ref{lem:mac-cap2}, it now follows that any body of $\RR$ that overlaps $E_j^{\sigma} \cap L_j$ is contained within the cap $E_j^{2 \sigma}$. By Lemma~\ref{lem:cap-exp}, $\vol(E_j^{2 \sigma}) = O(\vol(E_j))$. Since the bodies of $\RR$ in layer $j$ have volumes $\Omega(\vol(E_j))$, it follows by a simple packing argument that at most a constant number of bodies of $\RR$ are contained within $E_j^{2 \sigma} \cap L_j$. Hence, the number of bodies of $\RR$ that overlap $E_j^{\sigma} \cap L_j$ is $O(1)$.

Next we bound the number of bodies of $\RR$ that overlap $E_r^{\sigma} \cap L_r$, where $j+1 \le r \le t$ (see Figure~\ref{f:prop3}). By Lemma~\ref{lem:cap-vol-poly}, we have $\vol(E_r) = O((r-j)^{3d}) \cdot \vol(E_j)$. Recall that the volume of the bodies of $\RR'$ in group $r$ exceeds the volume of the bodies of $\RR'$ in group $j$ by a factor of $\Omega(2^{r-j})$. It follows from Lemma~\ref{lem:scale}\ref{e} and our construction that the volume of the bodies of $\RR$ in layer $r$ exceeds the volume of the bodies of $\RR$ in layer $j$ by a factor of $\Omega(2^{r-j})$. For the same reasons as discussed above, any body of $\RR$ that overlaps $E_r^{\sigma} \cap L_r$ is contained within $E_r^{2\sigma}$, and $\vol(E_r^{2\sigma}) = O(\vol(E_r))$. Putting this together with the upper bound on $\vol(E_r)$ given above, we have $\vol(E_r^{2 \sigma}) = O((r-j)^{3d}) \cdot \vol(E_j)$. By a simple packing argument, it follows that the number of bodies of $\RR$ that are contained within $E_r^{2\sigma} \cap L_r$ is $O((r-j)^{3d} / 2^{r-j})$. This bounds the number of bodies of $\RR$ that overlap $E_r^{\sigma} \cap L_r$. It follows that the number of bodies of $\RR$ that overlap $C =  \bigcup_{r = j+1}^{t} (E_r^{\sigma} \cap L_r)$ is on the order of $\sum_{j+1 \le r \le t} (r-j)^{3d} / 2^{r-j} = O(1)$, which completes the proof.
\end{proof}

\subsection{Polytope Approximation}

Finally, we can assemble all the pieces to obtain the desired approximation. Let $S$ be a set of points containing one point inside each body of $\RR$ defined in Lemma~\ref{lem:layers} and no other point.

\begin{lemma} \label{lem:apx}
The polytope $P = \conv(S)$ is an inner $\eps$-approximation of $K$.
\end{lemma}

\begin{proof}
A set of points $S$ is said to \emph{stab} a cap if the cap contains at least one point of $S$. It is well known that if a set of points $S \subset K$ stabs all caps of width $\eps$ of $K$, then $\conv(S)$ is an inner $\eps$-approximation of $K$~\cite{BrI76}. Let $C$ be a cap of width $\eps$. By Lemma~\ref{lem:layers}, Property~{1}, there is a convex body $R_i \subseteq C$. Since $S$ contains a point that is in $R_i$, we have that the cap $C$ is stabbed.
\end{proof}

To bound the combinatorial complexity of $\conv(S)$, and hence conclude the proof of Theorem~\ref{thm:main}, we use the witness-collector approach~\cite{DGGT16}. The proof is the analogue to the one in~\cite{AFM17a} and is included for completeness.

\begin{lemma} \label{lem:fewfaces}
The number of faces of $P=\conv(S)$ is $O(1/\eps^{(d-1)/2})$.
\end{lemma}

\begin{proof}
Define the witness set $\WW = R_1, \ldots, R_k$ and the collector set $\CC = C_1, \ldots, C_k$, where the $R_i$'s and $C_i$'s are as defined in Lemma~\ref{lem:layers}. As there is a point of $S$ in each body $R_i$, Property~(1) of the witness-collector method is satisfied. To prove Property~(2), let $H$ be any halfspace. If $H$ does not intersect $K$, then Property~(2) of the witness-collector method holds trivially. Otherwise let $C = K \cap H$. By Property~2 of Lemma~\ref{lem:layers}, there is an $i$ such that either $R_i \subseteq C$ or $C \subseteq C_i$. It follows that $H$ contains witness $R_i$ or $H \cap S$ is contained in collector $C_i$. Thus Property~(2) of the witness-collector method is satisfied. Finally, Property~3 of Lemma~\ref{lem:layers} implies Property~(3) of the witness-collector method. Thus, we can apply Lemma~\ref{lem:witness-collector} to conclude that the number of faces of $P$ is $O(|\CC|) = O(k)$, which proves the lemma.
\end{proof}

\section{Conclusion and Open Problems} \label{s:conclusion}

We have shown that it is possible to $\eps$-approximate any convex body $K \subset \RE^d$ by a polytope of combinatorial complexity $O(1/\eps^{(d-1)/2})$. This bound removes the polylogarithmic factor of our previous construction~\cite{AFM17a} and is the best possible since a unit ball requires this number of vertices. This improved bound arises from a deeper understanding of the distribution of Macbeath regions and caps depending on volume. In particular, we have shown that the number of disjoint $\eps$-width Macbeath regions of volume $v$ is $O\kern-1pt\left(\min\left(\frac{\eps}{v}, \frac{v}{\eps^d}\right)\right)$. This volume-sensitive bound led to a more efficient application of the witness-collector method.

A number of open problems remain. The first is how to efficiently compute the approximating polytope. By applying the techniques from~\cite{AFM17c}, it is possible to compute the vertices of the approximation in time $O(n + 1/\eps^{3d/2})$. Can this be improved to match time needed to compute an $\eps$-kernel? The current best bound is $O(n \log\inv\eps + 1/\eps^{(d-1)/2 + \alpha})$, for any constant $\alpha > 0$~\cite{AFM17c}. Extending this, is it possible to compute not only the vertices but the whole combinatorial description of $P$ in this time? Note that, by Avis and Fukuda's convex hull algorithm~\cite{AvF92}, the combinatorial description can be computed from the vertices in $O(1/\eps^{d-1})$ time. Again, we may ask if the exponent may be reduced to roughly $(d-1)/2$.

Another valuable direction for future research is that of simplifying the construction. An important question in this regard is whether the constructions of Dudley or Bronshteyn and Ivanov lead to polytopes of optimal combinatorial complexity. Building mathematical tools that would allow us to analyze the complexity of these constructions is a major open problem. With these tools in hand, it would be desirable to analyze the combinatorial complexity of the convex hull of the point sets produced by the fastest $\eps$-kernel construction algorithms known~\cite{AFM17c,Cha18}.

As mentioned earlier, our result is optimal in the worst case, but it may be suboptimal for particular convex bodies. Clarkson~\cite{Cla93} presented an algorithm to compute an $\eps$-approximation of a given convex body $K$ using at most $O(\log\inv\eps)$ times the minimum number of vertices needed to $\eps$-approximate $K$. (Clarkson's result applies to the $\ell_1$, rather than Euclidean, metric.) We wonder if a similar approximation algorithm exists for the total combinatorial complexity, or, failing this, whether hardness results can be proved. A related problem involves obtaining bounds on the combinatorial complexity that depend on certain parameters of the input polytope, such as its surface area, volume, or number of facets.

Throughout, we have assumed that the dimension $d$ is a constant. A careful examination of our proofs reveals factors of the form $d^{O(d)} = 2^{O(d \log d)}$ hidden in our asymptotic forms for the number of witnesses and collectors. Since a collector may intersect $d^{O(d)}$ witnesses, the witness-collector technique produces a factor of $d^{O(d^2)}$ in the combinatorial complexity. In contrast, the $d$-dependency in the number of vertices obtained by Bronshteyn and Ivanov's construction is of the form $2^{O(d)}$. Another natural question is whether this gap can be closed. In fact, even for a constant $\eps$, it is not known whether there exists an approximating polytope of combinatorial complexity $2^{O(d)}$.

\section{Acknowledgements} \label{s:acks}

We would like to thank the reviewers of both the conference and journal versions of this paper for their many helpful suggestions.

\bibliographystyle{abbrv}
\bibliography{shortcuts,macbeath}

\begin{thebibliography}{10}

\bibitem{ANS16}
K.~Adiprasito, E.~Nevo, and J.~A. Samper.
\newblock A geometric lower bound theorem.
\newblock {\em Geom.\ Funct.\ Anal.}, 26:359--378, 2016.

\bibitem{AHV05}
P.~K. Agarwal, S.~Har-Peled, and K.~R. Varadarajan.
\newblock Geometric approximation via coresets.
\newblock In J.~E. Goodman, J.~Pach, and E.~Welzl, editors, {\em Combinatorial
  and Computational Geometry}. MSRI Publications, Berkeley, CA, 2005.

\bibitem{And63}
G.~E. Andrews.
\newblock A lower bound for the volumes of strictly convex bodies with many
  boundary points.
\newblock {\em Trans.\ Amer.\ Math.\ Soc.}, 106:270--279, 1963.

\bibitem{AFM12b}
S.~Arya, G.~D. da~Fonseca, and D.~M. Mount.
\newblock Optimal area-sensitive bounds for polytope approximation.
\newblock In {\em Proc.\ 28th Annu.\ Sympos.\ Comput.\ Geom.}, pages 363--372,
  New York, NY, 2012. Association for Computing Machinery.

\bibitem{AFM12a}
S.~Arya, G.~D. da~Fonseca, and D.~M. Mount.
\newblock Polytope approximation and the {Mahler} volume.
\newblock In {\em Proc.\ 23rd Annu.\ ACM-SIAM Sympos.\ Discrete Algorithms},
  pages 29--42, Philadelphia, PA, 2012. Society for Industrial and Applied
  Mathematics.

\bibitem{AFM17c}
S.~Arya, G.~D. da~Fonseca, and D.~M. Mount.
\newblock Near-optimal $\varepsilon$-kernel construction and related problems.
\newblock In {\em Proc.\ 33rd Internat.\ Sympos.\ Comput.\ Geom.}, pages
  10:1--10:15, Dagstuhl, Germany, 2017. Schloss Dagstuhl--Leibniz-Zentrum f{\"
  u}r Informatik.

\bibitem{AFM17a}
S.~Arya, G.~D. da~Fonseca, and D.~M. Mount.
\newblock On the combinatorial complexity of approximating polytopes.
\newblock {\em Discrete Comput.\ Geom.}, 58:849--870, 2017.

\bibitem{AFM17b}
S.~Arya, G.~D. da~Fonseca, and D.~M. Mount.
\newblock Optimal approximate polytope membership.
\newblock In {\em Proc.\ 28th Annu.\ ACM-SIAM Sympos.\ Discrete Algorithms},
  pages 270--288, Philadelphia, PA, 2017. Society for Industrial and Applied
  Mathematics.

\bibitem{AFM18}
S.~Arya, G.~D. da~Fonseca, and D.~M. Mount.
\newblock Approximate convex intersection detection with applications to width
  and {Minkowski} sums.
\newblock In {\em Proc.\ 26th Annu.\ European Sympos.\ Algorithms}, pages
  3:1--14, Dagstuhl, Germany, 2018. Schloss Dagstuhl--Leibniz-Zentrum f{\" u}r
  Informatik.

\bibitem{AMM09b}
S.~Arya, T.~Malamatos, and D.~M. Mount.
\newblock The effect of corners on the complexity of approximate range
  searching.
\newblock {\em Discrete Comput.\ Geom.}, 41:398--443, 2009.

\bibitem{AMX12}
S.~Arya, D.~M. Mount, and J.~Xia.
\newblock Tight lower bounds for halfspace range searching.
\newblock {\em Discrete Comput.\ Geom.}, 47:711--730, 2012.

\bibitem{AvF92}
D.~Avis and K.~Fukuda.
\newblock A pivoting algorithm for convex hulls and vertex enumeration of
  arrangements and polyhedra.
\newblock {\em Discrete Comput.\ Geom.}, 8(3):295--313, 1992.

\bibitem{Bar89}
I.~B{\'a}r{\'a}ny.
\newblock Intrinsic volumes and $f$-vectors of random polytopes.
\newblock {\em Math.\ Ann.}, 285:671--699, 1989.

\bibitem{Bar00}
I.~B{\'a}r{\'a}ny.
\newblock The technique of {M}-regions and cap-coverings: {A} survey.
\newblock {\em Rend.\ Circ.\ Mat.\ Palermo}, 65:21--38, 2000.

\bibitem{Bar08}
I.~B{\'a}r{\'a}ny.
\newblock Extremal problems for convex lattice polytopes: {A} survey.
\newblock {\em Contemp.\ Math.}, 453:87--103, 2008.

\bibitem{BKT17}
J.-D. Boissonnat, {Karthik {C. S.}}, and S.~Tavenas.
\newblock Building efficient and compact data structures for simplicial
  complexes.
\newblock {\em Algorithmica}, 79(2):530--567, 2017.

\bibitem{Bor00a}
K.~B{\" o}r{\" o}czky~Jr.
\newblock Approximation of general smooth convex bodies.
\newblock {\em Adv.\ Math.}, 153:325--341, 2000.

\bibitem{Bor00b}
K.~B{\"o}r{\"o}czky~Jr.
\newblock Polytopal approximation bounding the number of $k$-faces.
\newblock {\em J.\ Approx.\ Theory}, 102(2):263--285, 2000.

\bibitem{BCP93}
H.~Br{\"o}nnimann, B.~Chazelle, and J.~Pach.
\newblock How hard is halfspace range searching.
\newblock {\em Discrete Comput.\ Geom.}, 10:143--155, 1993.

\bibitem{BrI76}
E.~M. Bronshteyn and L.~D. Ivanov.
\newblock The approximation of convex sets by polyhedra.
\newblock {\em Siberian Math.\ J.}, 16:852--853, 1976.

\bibitem{Bro08}
E.~M. Bronstein.
\newblock Approximation of convex sets by polytopes.
\newblock {\em J.\ Math.\ Sci.}, 153(6):727--762, 2008.

\bibitem{Cha18}
T.~M. Chan.
\newblock Applications of {Chebyshev} polynomials to low-dimensional
  computational geometry.
\newblock {\em J.\ Comput.\ Geom.}, 9(2):3--20, 2018.

\bibitem{Cla93}
K.~L. Clarkson.
\newblock Algorithms for polytope covering and approximation.
\newblock In {\em Proc.\ Third Internat.\ Workshop Algorithms Data Struct.},
  pages 246--252, Berlin, 1993. Springer.

\bibitem{Cla06}
K.~L. Clarkson.
\newblock Building triangulations using $\varepsilon$-nets.
\newblock In {\em Proc.\ 38th Annu.\ ACM Sympos.\ Theory Comput.}, pages
  326--335, New York, NY, 2006. Association for Computing Machinery.

\bibitem{DGGT16}
O.~Devillers, M.~Glisse, X.~Goaoc, and R.~Thomasse.
\newblock Smoothed complexity of convex hulls by witnesses and collectors.
\newblock {\em J.\ Comput.\ Geom.}, 7(2):101--144, 2016.

\bibitem{Dud74}
R.~M. Dudley.
\newblock Metric entropy of some classes of sets with differentiable
  boundaries.
\newblock {\em J.\ Approx.\ Theory}, 10(3):227--236, 1974.

\bibitem{DGJ19}
K.~Dutta, A.~Ghosh, B.~Jartoux, and N.~H. Mustafa.
\newblock Shallow packings, semialgebraic set systems, {Macbeath} regions and
  polynomial partitioning.
\newblock {\em Discrete Comput.\ Geom.}, 61:756--777, 2019.

\bibitem{Egg58}
H.~G. Eggleston.
\newblock {\em Convexity}.
\newblock Cambridge University Press, Cambridge, UK, 1958.

\bibitem{ELR70}
G.~Ewald, D.~G. Larman, and C.~A. Rogers.
\newblock The directions of the line segments and of the $r$-dimensional balls
  on the boundary of a convex body in {E}uclidean space.
\newblock {\em Mathematika}, 17:1--20, 1970.

\bibitem{Gru93}
P.~M. Gruber.
\newblock Asymptotic estimates for best and stepwise approximation of convex
  bodies {I}.
\newblock {\em Forum Math.}, 5:521--537, 1993.

\bibitem{John}
F.~John.
\newblock Extremum problems with inequalities as subsidiary conditions.
\newblock In {\em Studies and Essays Presented to R. Courant on his 60th
  Birthday}, pages 187--204. Interscience Publishers, Inc., New York, 1948.

\bibitem{KLS95}
R.~Kannan, L.~Lov{\'a}sz, and M.~Simonovits.
\newblock Isoperimetric problems for convex bodies and a localization lemma.
\newblock {\em Discrete Comput.\ Geom.}, 13(3-4):541--559, 1995.

\bibitem{Kuperberg}
G.~Kuperberg.
\newblock From the {Mahler} conjecture to {Gauss} linking integrals.
\newblock {\em Geom.\ Funct.\ Anal.}, 18:870--892, 2008.

\bibitem{MCM70}
P.~McMullen.
\newblock The maximum numbers of faces of a convex polytope.
\newblock {\em Mathematika}, 17:179--184, 1970.

\bibitem{MuR14}
N.~H. Mustafa and S.~Ray.
\newblock Near-optimal generalisations of a theorem of {Macbeath}.
\newblock In {\em Proc.\ 31st Internat.\ Sympos.\ on Theoret.\ Aspects of
  Comp.\ Sci.}, pages 578--589, Dagstuhl, Germany, 2014. Schloss
  Dagstuhl--Leibniz-Zentrum f{\" u}r Informatik.

\bibitem{Sch87}
R.~Schneider.
\newblock Polyhedral approximation of smooth convex bodies.
\newblock {\em J.\ Math.\ Anal.\ Appl.}, 128:470--474, 1987.

\bibitem{Tho17}
N.~Tholozan.
\newblock Volume entropy of {Hilbert} metrics and length spectrum of {Hitchin}
  representations into {PSL(3,R)}.
\newblock {\em Duke Math. J.}, 166(7):1377--1403, 2017.

\bibitem{Tot48}
L.~F. Toth.
\newblock Approximation by polygons and polyhedra.
\newblock {\em Bull.\ Amer.\ Math.\ Soc.}, 54:431--438, 1948.

\bibitem{VeW16}
C.~Vernicos and C.~Walsh.
\newblock Flag-approximability of convex bodies and volume growth of {H}ilbert
  geometries.
\newblock HAL Archive (hal-01423693i), 2016.

\end{thebibliography}

\end{document}